\documentclass[11pt]{article}

\usepackage{amsmath}
\usepackage{amssymb}

\usepackage{amsthm}
\usepackage{verbatim}
\usepackage{graphicx}
\usepackage{epsfig}
\usepackage{url}
\usepackage{float}
\usepackage[all]{xy}
\usepackage{color}

\usepackage{fullpage}
\usepackage{appendix}
\usepackage[nomessages]{fp}

\newtheorem{theorem}{Theorem}[section]

\newtheorem{claim}[theorem]{Claim}

\newtheorem{corollary}[theorem]{Corollary}
\newtheorem{lemma}[theorem]{Lemma}

\newtheorem{openproblem}[theorem]{Open Problem}

\theoremstyle{remark}
\newtheorem{remark}[theorem]{Remark}
\theoremstyle{definition}
\newtheorem{definition}[theorem]{Definition}

\newenvironment{prf}{\noindent{\bf Proof.~}}{\(\qed\)}
\newcommand{\BPF}{\begin{prf}} \newcommand {\EPF}{\end{prf}}

\newcommand{\jiemingnote}[1]{{\color{red}{#1}}}

\newcommand{\ignore}[1]{{}}

\newcommand{\ket}[1]{\mathop{\left|#1\right>}\nolimits}
\newcommand{\bra}[1]{\mathop{\left<#1\,\right|}\nolimits}

\newcommand{\kb}[2]{| #1\rangle\!\langle #2 |}
\newcommand{\Tra}[1]{\mathop{{\mathrm{Tr}}_{#1}}}
\newcommand{\Tr}[2]{\mathop{{\mathrm{Tr}}_{#1}} (#2) }

\newcommand{\eps}{\epsilon}

\newcommand{\E}{\bf E\rm}

\def\N{\mathcal{N}}

\def\T{\mathcal{T}}
\def\C{\mathcal{C}}

\def\D{\mathcal{D}}

\def\U{\mathcal{U}}

\def\polylog{\text{polylog}}

\def\T{\mathcal{T}}

 \allowdisplaybreaks

\floatstyle{boxed}
\newfloat{Protocol}{h}{pro}

\begin{document}

\title{Near-optimal bounds on bounded-round quantum communication complexity of disjointness}

\author{
Mark Braverman \thanks{Department of Computer Science, Princeton University, email: mbraverm@cs.princeton.edu. Research supported in part by an NSF CAREER award (CCF-1149888), a 
Turing Centenary Fellowship, a Packard Fellowship in Science and Engineering, and the Simons Collaboration on Algorithms and Geometry.}
\and
Ankit Garg \thanks{Department of Computer Science, Princeton University, email: garg@cs.princeton.edu. Research supported by a Simons Fellowship in Theoretical Computer Science.}
\and
Young Kun Ko  \thanks{Department of Computer Science, Princeton University, email: yko@cs.princeton.edu}
\and
Jieming Mao  \thanks{Department of Computer Science, Princeton University, email: jiemingm@cs.princeton.edu}
\and 
Dave Touchette \thanks{D\'epartement d'informatique et de recherche op\'erationnelle, Universit\'e de Montr\'eal, email: touchette.dave@gmail.com Research supported in part by an FRQNT B2 Doctoral Research Scholarship and by CryptoWorks21.}
}

\maketitle

\begin{abstract}
We prove a near optimal round-communication tradeoff for the two-party quantum communication complexity of disjointness. For protocols with $r$ rounds, we prove a lower bound of $\tilde{\Omega}(n/r + r)$ on the communication required for computing disjointness of input size $n$, which is optimal up to logarithmic factors. The previous best lower bound was $\Omega(n/r^2 + r)$ due to Jain, Radhakrishnan and Sen \cite{jrs_disj}. Along the way, we develop several tools for quantum information complexity, one of which is a lower bound for quantum information complexity in terms of the generalized discrepancy method. As a corollary, we get that the quantum communication complexity of any boolean function $f$ is at most $2^{O(QIC(f))}$, where $QIC(f)$ is the prior-free quantum information complexity of $f$ (with error $1/3$).  
\end{abstract}

\newpage
\tableofcontents

\pagebreak

\section{Introduction}

%Communication complexity, introduced by Yao in 1979 \cite{Yao79}, studies the amount of communication that two parties, Alice and Bob, need to exchange in order to compute a function (usually boolean) of their private inputs, $x$ and $y$, respectively. It has been a ubiquitous tool in proving unconditional lower bounds for several computational models such as
%VLSI design, streaming algorithms, data structures, property testing, branching programs, proof complexity, etc. Furthermore, some of the best approaches for proving circuit lower bounds, such as Karchmer-Wigderson games \cite{KarchmerRW95} and ACC lower bounds \cite{BeigelT91}, rely on communication complexity. 

We prove near-optimal bounds on the bounded-round quantum communication complexity of disjointness. Quantum communication complexity, introduced by Yao \cite{Qcc_Yao}, studies the amount of quantum communication that two parties, Alice and Bob, need to exchange in order to compute a function (usually boolean) of their private inputs. It is the natural quantum extension of classical communication complexity \cite{Yao79}.
While the inputs are classical and the end result is classical, the players are allowed to use quantum resources while communicating. The motivation for the introduction of quantum communication was to study questions in quantum computation. For example, in \cite{Qcc_Yao}, Yao used it to prove that the majority function does not have any linear size quantum formulas. 

While quantum communication (with entanglement) offers only a factor of $2$ savings when transmitting $n$ bits of classical information \cite{holevo,superdensecoding,hol_entanglement}, it can still offer super-constant savings (and sometimes exponential) in communication if the goal is just to compute a boolean function of the inputs. For total boolean functions, the best-known separation between classical and quantum communication is quadratic, for the disjointness function \cite{KalyanasundaramS92,Razborov92,grover, BCW, AA03}. It is, in fact, a major open problem whether classical and quantum communication are polynomially related for all total boolean functions. For partial functions, exponential separations are known even between one-way quantum communication and arbitrary classical communication \cite{Raz99,RegevK11}.
 
  For disjointness with input size $n$, Grover's search \cite{grover, BBHT98} can be used to obtain a quantum communication protocol (with probability of error $1/3$) with communication cost $O(\sqrt{n} \log n)$ \cite{BCW}. The bound was later improved to $O(\sqrt{n})$ in \cite{AA03}. The protocols attaining this upper bound are very interactive and require $\Theta(\sqrt{n})$ rounds of interaction. The $O(\sqrt{n})$ upper bound on the quantum communication complexity of disjointness has been shown to be tight in \cite{razborov2003}. 
	
	If we restrict the players to allow only $r$ rounds of interaction, then it is not hard to use the $O(\sqrt{n})$ protocol discussed above as a black-box to obtain an $O(n/r)$ communication protocol for $n\ge r^2$. The best known lower bound was $\Omega(n/r^2)$ \cite{jrs_disj}. We prove a lower bound of $\tilde{\Omega}(n/r)$, which is optimal up to logarithmic factors: 
	
	\smallskip
	\noindent
	{\bf Theorem A.} (Theorem~\ref{thm:disjround}, rephrased) {\em The $r$-round quantum communication complexity of $DISJ_n$ is $\Omega\left(\frac{n}{r \log^8(r)}\right)$.
	}
	\smallskip
	
	The analogous result for query complexity of quantum search, an $\Omega(n/r)$ lower bound for the number of queries when $r$ sets of nonadaptive queries are allowed, was known before \cite{zalka_groversearch}. Our lower bound does not give a new proof of the $\Omega(\sqrt{n})$ bound on the quantum communication complexity of disjointness \cite{razborov2003} since our proof uses that lower bound (in fact we use something much stronger, a strengthening of the strong direct product theorem for disjointness \cite{klauck2004quantum} due to \cite{sdp_sherstov}).
 
 There is a rich history of papers studying lower bounds on bounded-round communication complexity, for example for the pointer jumping problem \cite{NW_pointer,PRV2001,Klauck1998,KNTZ2001}, for sparse set disjointness \cite{ST2013disj}, for equality \cite{BCKequality} and several other examples. Most of these lower bounds are proven via a round elimination strategy: show that an $r$-round protocol can be converted into an $(r-1)$-round protocol without too much increase in communication cost and error; arrive at contradition by obtaining a too-good-to-be-true $1$-round or $0$-round protocol. Even the result of \cite{jrs_disj} can be viewed as round elimination on quantum information complexity of the $2$-bit AND. Despite substantial effort, obtaining the optimal $\Omega(1/r)$ lower bound on the  $r$-round quantum information complexity of AND via round elimination has remained elusive. We prove:

	\smallskip
	\noindent
	{\bf Theorem B.} (Corollary~\ref{cor:icand}, rephrased) {\em The $r$-round quantum information complexity of AND with prior $1/3,1/3,1/3,0$ is $\Omega\left(\frac{1}{r \log^8(r)}\right)$. 
	}
	\smallskip

As discussed below, we obtain this result by using existing lower bounds for the communication complexity of quantum disjointness. A direct proof of a quantum information complexity lower bound for the $2$-bit AND remains an intriguing open problem. In light of the fact that disjointness has a sub-linear quantum communication complexity, it is not surprising that the quantum information complexity of AND vanishes with the number of rounds. This phenomenon is closely related to the Elitzur-Vaidman bomb tester \cite{elitzur1993quantum, bomb_testing}, which gives a sequence of quantum measurements that allows one to test whether a bomb is loaded without detonating it. The loss of the protocol (i.e. the probability that the bomb will explode --- which loosely corresponds to the amount of information revealed about the bomb) behaves like $1/r$, where $r$ is the number of measurements performed.

 %Another scenario in which similar questions are studied is the CONGEST model of computation in distributed computing. The model considers a large number of players who want to compute some joint function of their inputs while communicating on a communication network. In one round, each player can send messages to its neighbors, but has a bandwidth of $B$ bits i.e. the total number of bits each player can send is at most $B$. The quantity of interest is how many rounds are needed to compute a specific function. For example, see \cite{congest} and the references there in. Note that for two parties, round-communication tradeoff and limited-bandwidth computation problems are almost the same. For example, we can use our round-communication tradeoff to obtain a lower bound on the number of rounds need to compute disjointness in a model of limited-bandwidth quantum computation. Suppose in each round, both the players are allowed to exchange messages of $B$ qubits each. Then if disjointness can be computed in $r$ rounds, total communication is $2 B  r$, and by our round-communication tradeoff, we get that $2 B r \ge \tilde{\Omega}(n/r)$, which gives $r \ge \tilde{\Omega}(\sqrt{n/B})$. This is optimal up to logarithmic factors, since the players can divide their inputs into $B$ equal sized blocks and run the disjointness protocol of \cite{AA03} on each of these blocks in parallel. Also the lower bound on total communication $\Omega(\sqrt{n})$ for disjointness would only give a lower bound of $\Omega(\sqrt{n}/B)$ on the round complexity.
 
 Our proof relies on the notion of quantum information complexity, defined recently in \cite{Tou14}, where it is used to prove a direct sum theorem for constant round quantum communication. It is harder to manipulate quantum information than in the classical case, and tools that are standard in the classical setting are yet to be developed for the quantum case. However, it could still be useful in proving partial direct sum and direct product theorems, which we know in the classical world \cite{BBCR10}, \cite{BRWY2013_dp}.  Moreover, a model similar to that of quantum communication complexity is connected to proving SDP extension complexity lower bounds \cite{JainSWZ13}. Although the recent breakthrough for SDP lower bounds \cite{LRS_sdp} does not follow this direction, it is likely that a quantum information complexity viewpoint will provide further insights as information complexity has provided in the classical case (LP extension complexity) \cite{BM12,BraunPokutta}. Further development of tools
for quantum communication and information complexity is likely to further the SDP extension complexity program.  
 
 We also prove that for all boolean functions, prior-free quantum information complexity is lower bounded by the generalized discrepancy method:

\smallskip
	\noindent
	{\bf Theorem C.} (Theorem~\ref{lem:QICvsGDM}, rephrased) {\em For any boolean function $f$ and a sufficiently small constant error $\eta > 0$, the prior-free quantum information complexity of $f$ with error $\eta$ is lower bounded by the generalized discrepancy bound for $f$. 
	}
	\smallskip

 Previously no lower bounds were known on the quantum information complexity of general boolean functions. Our proof relies on the strong direct product theorem for quantum communication complexity in terms of the generalized discrepancy method \cite{sdp_sherstov}. Note that in the classical setting such a result can be proven directly using zero-communication protocols~\cite{kerenidis2012lower}. It remains to be seen whether such a direct proof can be obtained in the quantum setting. 
 
 As a corollary we also get that the quantum communication complexity of any boolean function is at most exponential in the prior-free quantum information complexity.

 \smallskip
	\noindent
	{\bf Theorem D.} (Corollary~\ref{qcc:qic}, rephrased) {\em For any boolean function $f$, quantum communication complexity of $f$ with error $1/3$ is at most $2^{O(QIC(f,1/3) + 1)}$, where $QIC(f,1/3)$ is the prior-free quantum information complexity of $f$ with error $1/3$.
	}
	\smallskip
	
Note that the classical analogue of this is proven via a compression argument	 \cite{B12}, but we prove this via an indirect argument. It would be interesting to prove this directly via a quantum compression argument. 
	
\subsection*{Acknowledgments} 

We would like to thank Andris Ambainis, Rahul Jain, Ashwin Nayak, Jaikumar Radhakrishnan, Iordanis Kerenidis and Mathieu Lauriere for helpful discussions.

\section{Proof overview and discussion}

\paragraph{High-level strategy.} At a high-level, the proof builds on the connection between quantum information complexity and quantum communication complexity of the disjointness function $DISJ_m$ with various values of $m$. 
There are two parts to the proof: 
\begin{enumerate}
\item Suppose there is a $r$-round quantum protocol for disjointness of input size $n\ge r^2$ with communication cost $\frac{n}{r \cdot \polylog(r)}$. Then there exists a protocol for disjointness of input size $r^2$ with quantum information cost $\le o(r)$. 
\item Lower bound on quantum information complexity of disjointness: we prove that the (prior-free) quantum information complexity of any boolean function is lower bounded by the generalized discrepancy method, which by results in \cite{patternmatrix} implies that quantum information complexity of disjointness with input size $r^2$ is $\Omega(r)$. 
\end{enumerate}

\noindent Note that these two steps imply a lower bound on the bounded round quantum communication complexity of disjointness. Also the above statements are about computation with some constant error (say $1/3$). 

Both directions are proven via a connection between the information complexity of a problem and its communication complexity. In one direction, a protocol for a large sized disjointness can be converted into a low-information protocol for a smaller size disjointness. Using  the converse direction of the connection, a low-information protocol for $DISJ_{r^2}$ leads to a protocol for many copies of the problem that violate known direct product results. The former connection has been at the heart of many classical lower bounds involving information complexity \cite{BaryossefJKS04, BGPW12}. The latter connection (deriving information complexity lower bound from known communication lower bound on an ``amortized" version of the problem) 
has been previously explored in the classical setting by \cite{braverman2013information}.

\medskip

Let us start by giving a high level overview of the first step. If there is a $r$-round quantum protocol for disjointness of input size $n$ with communication cost $\frac{n}{r \cdot \polylog(r)}$ and $1/3$ probability of error, then by a direct sum argument in \cite{Tou14}, there exists a $r$-round quantum protocol $\pi$ for AND with $1/3$ probability of error (for a worst case input) and quantum information cost $\le \frac{1}{r \cdot \polylog(r)}$ w.r.t \textbf{any distribution $\boldsymbol{\mu}$ s.t. $\boldsymbol{\mu(1,1) = 0}$}. Now we want to use $\pi$ to obtain a low information protocol for disjointness of size $r^2$. One can imagine if we run $\pi$ on each coordinate of the disjointness instance, we get an $r$-round protocol $\tau$ of information cost $\le  \frac{r}{\polylog(r)}$ and also it solves disjointness with small error (assuming we first amplify the error of $\pi$ to $1/r^3$ losing a log factor in information cost). However, the issue is that information cost of $\tau$ is low only w.r.t. \textbf{distributions $\boldsymbol{\nu}$ supported on disjoint pairs of sets}. The information cost of $\tau$ may increase dramatically when it is run on a pair of sets with many intersections. To deal with this we use a trick used in \cite{BGPW12}. 

Note that if there are too many intersections in a disjointness instance, then the players can just subsample some of the coordinates and check for an intersection in those coordinates. Hence we can assume wlog that the intersection size in a typical input distributed according to $\nu$ is small. This means that if we look at a typical coordinate $i$, the marginal distribution $\nu_i$ has small mass on $(1,1)$. And in this case, we can run $\pi$ on each coordinate. The only thing left to understand is: how does the information cost of $\pi$ change if we place a small mass, say $w$, on $(1,1)$? The answer to this turns out to be $r\cdot H(w)$, where $\pi$ has $r$-rounds. Note that this is in contrast to the classical case, where the answer would be just $H(w)$. Later we will give an example of a quantum protocol for AND whose information cost does go up by $r\cdot H(w)$. Also \textbf{this is the only place where we use the fact that the protocol we started with had only $\boldsymbol{r}$ rounds}. Such a dependence is necessary here, since an $\Omega(n/r)$ lower bound for general (non-$r$-round) protocols would 
violate the $O(\sqrt{n})$ upper bound. 

\medskip

For the second step, we use compression along with a strong direct product theorem for quantum communication complexity of $f$ in terms of the generalized discrepancy lower bound $GDM_{1/5} (f)$ due to Sherstov \cite{sdp_sherstov}. It says that to compute $k$ copies of a boolean function $f$ with success probability $2^{-\Omega(k)}$, it requires at least $k \cdot GDM_{1/5}(f)$ qubits of communication (with arbitrary amount of entanglement). Note that a strong direct product theorem for quantum communication complexity of disjointness was already known \cite{klauck2004quantum}, but we need a stronger version for our proof which shows that even computing a large fraction of the copies is hard and Sherstov's result also holds in this case\footnote{We could probably 
base our result off the lower bound of \cite{klauck2004quantum}, but the reduction would be considerably 
more complicated.}.

 Suppose there is a protocol $\pi$ for a function $f$ with quantum information cost $\le I$ w.r.t a distribution $\mu$ and probability of error $\le \eps$, then by quantum information equals amortized communication \cite{Tou14}, we get a protocol $\pi_k$ for $f^k$ which computes at least $(1-2\eps)k$ coordinates correctly with probability $\ge 0.99$ (w.r.t. $\mu^k$) and $QCC(\pi_k) \le k \cdot I + o(k)$. To apply Sherstov's theorem, we need such a protocol which works for worst case inputs. We show how to obtain such a worst case to average case reduction, whence applying Sherstov's result gives us the lower bound on information complexity.

\subsection*{Discussion and open problems}

In its entirety our proof shows how from a $r$-round protocol for disjointness, one can obtain a protocol for $k$ copies of disjointness of size $r^2$. But to achieve this reduction, we have to move to information complexity, since the number of rounds $r$ only comes up in an information theoretic context in our proof. 

Thus the reduction structure of the proof is communication$\rightarrow$information$\rightarrow$communication, with the latter communication problem having a known lower bound. Lower bounds for disjointness in the classical setting \cite{BaryossefJKS04,BGPW12} only do a reduction of the form 
communication $\rightarrow$ information, with an information complexity lower bound on the resulting problem proven directly. 

\begin{openproblem}
Give a direct proof of a lower bound for the information complexity of $DISJ_{r^2}$. 
\end{openproblem}

One possible attack route would be  along the lines of the proof for the classical case using zero-communication protocols \cite{kerenidis2012lower}. In the past, techniques developed for two-party quantum communication, e.g. the pattern matrix method \cite{patternmatrix}, turned out to be useful for multiparty number-on-forehead communication \cite{ChattopadhyayA08,sherstov_multiparty}. It could be that techniques developed for quantum information also result in similar progress.
 
Another natural question is whether the lower bound on the information complexity of AND can be proved using a direct argument:

\begin{openproblem} \label{AND_direct}
Give a direct proof of Theorem~B.
\end{openproblem}

Even though efforts since \cite{jrs_disj} to-date have been unsuccessful, it still could be possible to directly obtain Theorem~B via round elimination or other techniques and that would be really interesting, since it would also yield a new proof of the lower bound for quantum communication complexity of disjointness \cite{razborov2003,patternmatrix}. The recent breakthrough results in lower bounding conditional quantum mutual information \cite{fawzi-renner, BHOS_CQMI, BT_CQMI} should be relevant.

\begin{remark}
Our proofs can be adapted to show that the (unbounded round) zero-error quantum information complexity of AND w.r.t the prior $(1-\eps)/3,(1-\eps)/3,(1-\eps)/3,\eps$ is $\tilde{\Omega}(\sqrt{\eps})$. It is another intriguing question whether it is possible to have a direct proof for this. Note that this requires a global view of quantum information complexity, even though it is defined round by round. By a continuity argument this would also resolve open problem \ref{AND_direct}. 
\end{remark}

More generally, our understanding of the relationship between quantum information and communication complexity is in its early stages of development. Questions of interactive protocol compression occupy a central position in understanding the connection between classical information and communication complexity \cite{BBCR10,B12,ganor2014exponential}. In particular, \cite{BBCR10} shows that a protocol 
$\pi$ with information cost $I$ and communication cost $C$ can be compressed into a protocol with 
communication cost $\tilde{O}(\sqrt{I\cdot C})$. It remains open whether this (or an analagous) fact is true in the quantum setting: 

\begin{openproblem}
Given a quantum protocol $\pi$ over a distribution $\mu$ of inputs whose communication cost is $C$ and whose quantum information cost is $I$, can $\pi$ be simulated (with a small error) using a quantum protocol $\pi'$ whose communication cost is $\tilde{O}(\sqrt{I\cdot C})$?
\end{openproblem}

\section{Preliminaries}

%------------------------------------------------------------
\subsection{Quantum Information Theory}

We use the following notation for quantum theory; see \cite{Wat13, Wilde11} for more details. 
We associate a quantum register $A$ with
a corresponding vector space, also denoted by $A$. We only consider 
finite-dimensional vector spaces. A state of quantum register $A$ is 
represented by a density operator $\rho \in \D (A)$, with $\D (A)$ the set 
of all unit trace, positive semi-definite linear operators mapping $A$ into itself. 
We say that a state $\rho$ is pure if
it is a projection operator, i.e.~$(\rho^{A})^2 = \rho^{A} $. 
For a pure state $\rho$, we might use the pure state formalism, and represent $\rho$ by the
vector $\ket{\rho}$ it projects upon, i.e.~$\rho = \kb{\rho}{\rho}$; this is well-defined up to an irrelevant phase factor.

A quantum channel from quantum register $A$ into quantum register $B$ is represented
by a super-operator $\N^{A \rightarrow B} \in \C (A, B)$, with $\C (A, B ) $ the set of
 all completely positive, trace-preserving linear operators from $\D( A )$ 
into $\D(B)$. If $A=B$, we might simply write $\N^A$, and when systems are clear from context, we might drop the superscripts.
For channels $\N_1 \in \C (A, B), \N_2 \in \C (B, C)$, we denote their composition as $\N_2 \circ \N_1 \in \C (A, C)$, with action $(\N_2 \circ \N_1) (\rho) = \N_2 (\N_1 (\rho))$ on  any state $\rho \in \D (A)$. We might drop the $\circ$ symbol if the composition is clear from context.
For $A$ and $B$ isomorphic, we denote the identity mapping 
as $I^{A \rightarrow B}$, with some implicit choice for the change of basis.
 For $\N^{A_1 \rightarrow B_1} \otimes I^{A_2 \rightarrow B_2} \in \C(A_1 \otimes A_2, B_1 \otimes B_2)$, 
we might abbreviate this as $\N$ and leave the identity channel implicit 
when the meaning is clear from context. 

An important subset
 of $\C (A, B)$ when $A$ and $B$ are isomorphic spaces is the set
 of unitary channels $\U (A, B)$, the set of all maps $U \in\C (A, B)$ with an 
adjoint map $U^\dagger \in \C (B, A)$ such that
 $U^\dagger \circ U = I^{A}$ and $U \circ U^\dagger = I^{B}$. 
More generally, if $\dim (B) \geq \dim (A)$, we denote by $\U (A, B)$ the set of isometric channels, i.e.~the set of all maps
$V \in \C (A, B)$ with an adjoint map $V^\dagger \in \C (B, A)$ such that  $V^\dagger \circ V = I^{A}$.
Another important example of channel that we use is the
 partial trace $\Tr{B}{\cdot} \in C(A \otimes B, A)$ which effectively 
gets rid of the $B$ subsystem to obtain the marginal state on subsystem $A$. 
Fixing an orthonormal basis $\{ \ket{b} \}$ for $B$, we can write the action of $\Tra{B}$ on 
any $\rho^{AB} \in \D(A \otimes B)$ as $\Tr{B}{\rho^{AB}} = 
\sum_b \bra{b} \rho^{AB} \ket{b}$.
Note that the action of $\Tra{B}$ is independent of the choice of basis chosen to represent it,  so we unambiguously write $\rho^A = \Tr{B}{\rho^{AB}}$.
We also use the notation $\Tra{\neg A} = \Tra{B}$ to express that we want to keep only the $A$ register.

Fixing a basis also allows us to talk about classical states and 
joint states: $\rho \in \D(B)$ is classical (with respect to this basis) if it is diagonal in basis
 $\{ \ket{b} \}$, i.e.~$\rho = \sum_b p_B (b) \cdot \kb{b}{b}$ for some probability
 distribution $p_B$. More generally, subsystem $B$ of $\rho^{AB}$ is 
said to be classical if we can write $\rho^{AB} = \sum _b p_B (b) \cdot
\kb{b}{b}^B \otimes \rho_b^A$ for some $\rho_b^A \in \D(A)$.
An important example of a channel mapping a quantum system 
to a classical one is the measurement channel $\Delta_B$, defined 
as $\Delta_B ( \rho) = \sum_b \bra{b} \rho \ket{b} \cdot \kb{b}{b}^B$ for any $\rho \in \D(B)$.
Note that for any state $\rho \in \D (B_1 \otimes B_2 \otimes C \otimes R)$ of the form
\begin{align*}
	\ket{\rho}^{B_1 B_2 C R} = \sum_b \sqrt{p_B (b)} \cdot \ket{b}^{B_1} \ket{b}^{B_2} \ket{\rho_b}^{CR},
\end{align*}
we have $\Tr{B_2}{\rho^{B_1 B_2 C R}} = \sum_b p_B (b) \cdot \kb{b}{b}^{B_1} \otimes \rho_b^{CR}$
and $\Tr{B_2 R}{\rho^{B_1 B_2 C R}} = \sum_b p_B (b) \cdot \kb{b}{b}^{B_1} \otimes \rho_b^C$, with the state on $B_1$ classical in both cases.
Often, $A, B, C, \cdots$ will be used to discuss general systems,
 while $X, Y, Z, \cdots$ will be reserved for classical systems, or quantum systems like $B_1$ and $B_2$ above that are classical once one of them is traced out, and can be thought of as containing a quantum copy of the classical content of one another.

For a state $\rho^A \in \D (A)$, a purification is a pure state
 $\rho^{AR} \in \D (A \otimes R)$ satisfying $\Tr{R}{\rho^{AR}} = \rho^A$. 
If $R$ has dimension at least that of $A$, then such a purification always exists.
For a given $R$, all purifications are equivalent up to a unitary on $R$,
and more generally, if $\dim (R^\prime) \geq \dim (R)$ and $\rho_1^{AR}, \rho_2^{A R^\prime}$ are two purifications of $\rho^A$, then there exists an isometry $V_\rho^{R \rightarrow R^\prime}$ such that $\rho_2^{A R^\prime} = V_\rho (\rho_1^{A R})$.
For a channel $\N \in C(A, B)$, an isometric extension is a
 unitary $U_\N \in U(A , A^\prime \otimes B)$ with
 $\Tr{A^\prime}{U_\N (\rho^A )} 
= \N (\rho^A)$ for all $\rho^A$.
Such an extension always
 exists provided $A^\prime$ is of dimension at least $\dim (A)^2$.
For the measurement channel $\Delta_B$, an isometric extension is given by $U_\Delta = \sum_b \ket{b}^{B^\prime} \ket{b}^{B} \bra{b}^{B}$.

The notion of distance we use is the trace distance, defined for 
two states $\rho_1, \rho_2 \in \D (A)$ as the sum of the absolute values of the eigenvalues of their difference:
\begin{align*}
 \| \rho_1 - \rho_2 \|_A = \Tr{}{| \rho_1 - \rho_2 |}.
\end{align*}
It has an operational interpretation as four times the best bias possible in a 
state discrimination test between $\rho_1$ and $\rho_2$.
The subscript tells on which subsystems the trace distance is evaluated, and remaining subsystems
might need to be traced out.
We use the following results about trace distance. For proofs of 
these and other standard results in quantum information theory that 
we use, see \cite{Wilde11}. The trace distance is monotone under
 noisy channels: for any $\rho_1, \rho_2 \in \D (A)$ and $\N \in \C(A, B)$,
\begin{align}
	\| \N (\rho_1) - \N (\rho_2) \|_B \leq \| \rho_1 - \rho_2 \|_A.
\end{align}
For isometries, the inequality becomes an equality, a property called isometric
 invariance of the trace distance. Hence, for any $\rho_1, \rho_2 \in \D(A)$ and any $U \in \U (A, B)$, we have
\begin{align}
	\| U (\rho_1) - U (\rho_2) \|_B = \| \rho_1 - \rho_2 \|_A.
\end{align}
Also, the trace distance cannot be increased by adjoining an uncorrelated system:
for any $\rho_1, \rho_2 \in \D(A), \sigma \in \D (B)$
\begin{align}
	\| \rho_1 \otimes \sigma - \rho_2 \otimes \sigma \|_{AB} =  \| \rho_{1} - \rho_{2}  \|_{A}.
\end{align}
The trace distance obeys a property that we call joint linearity: 
for a classical system $X$ and two states $\rho_1^{XA} = 
p_X (x) \cdot \kb{x}{x}^X \otimes \rho_{1, x}^A$ and $\rho_2^{XA} = p_X (x) \cdot \kb{x}{x}^X \otimes \rho_{2, x}^A$,
 \begin{align}
	\| \rho_1 - \rho_2 \|_{XA} =  \sum_x p_X (x) \| \rho_{1,x} - \rho_{2, x}  \|_{A}.
\end{align}

The measure of information that we use is the von Neumann entropy, defined for any state $\rho \in \D(A)$ as
\begin{align*}
	H(A)_\rho = - \Tr{}{\rho \log \rho},
\end{align*}
in which we take the convention that $0 \log 0 = 0$, justified by a continuity argument. The logarithm $\log$ is taken in base $2$, while the natural logarithm is denoted $\ln$. 
Note that $H$ is invariant under isometries applied on $\rho$.
If the 
state to be evaluated is clear from context, we might drop the subscript. 
Conditional entropy for a state $\rho^{ABC} \in \D(A \otimes B \otimes C)$ is then defined as
\begin{align*}
	H(A | B) = H(AB) - H(B),
\end{align*}
mutual information as
\begin{align*}
	I(A; B ) = H (A) - H(A | B),
\end{align*}
and conditional mutual information as
\begin{align*}
	I(A; B | C ) = H (A | C) - H(A | B C).
\end{align*}
Note that mutual information and conditional mutual information are symmetric in interchange of $A, B$,
and invariant under a local isometry applied to $A, B$ or $C$.
For any pure bipartite state $\rho^{AB} \in \D(A \otimes B)$, the entropy on each subsystem is the same:
\begin{align}
	H(A) = H(B).
\end{align}
Since all purifications are equivalent up to an isometry on the purification registers, we get that for any two pure states
$\ket{\phi}^{ABC R^\prime}$ and $\ket{\psi}^{ABC R}$ such that 
$\phi^{ABC} = \psi^{ABC}$,
\begin{align}
\label{eq:pure}
	I(C;R'|B)_{\phi} = I(C;R|B)_{\psi}.
\end{align}

%\begin{fact} \label{fact:pure}
%If $\ket{\phi}^{R', A,B,C}$ and $\ket{\psi}^{R,A,B,C}$ are two pure states such that $\text{Tr}_{R'} \ket{\phi}^{R', A,B,C} = \text{Tr}_{R} \ket{\psi}^{R, A,B,C}$. Then $I(C;R'|B)_{\phi} = I(C;R|B)_{\psi}$.
%\end{fact}

%\begin{proof}
%The proof will use the following easy fact: if $\ket{\rho}^{E,F}$ is a pure state, then $H(E)_{\rho} = H(F)_{\rho}$. Now
%\begin{align*}
%I(C;R'|B)_{\phi} &= H(R',B)_{\phi} + H(C,B)_{\phi} - H(C,R',B)_{\phi} - H(B)_{\phi} \\
%&= H(A,C)_{\phi} + H(C,B)_{\phi} - H(A)_{\phi} - H(B)_{\phi} \\
%&= H(A,C)_{\psi} + H(C,B)_{\psi} - H(A)_{\psi} - H(B)_{\psi} \\
%&= H(R,B)_{\psi} + H(C,B)_{\psi} - H(R,C,B)_{\psi} - H(B)_{\psi} \\
%&= I(C;R|B)_{\psi}
%\end{align*}
%The first equality is by definition. Second is by the fact above. Third is from the fact that $\text{Tr}_{R'} \ket{\phi}^{R', A,B,C} = \text{Tr}_{R} \ket{\psi}^{R, A,B,C}$. Fourth is from the fact above. Fifth is by definition. 
%\end{proof}

For isomorphic $A, A^\prime$, a maximally entangled state $\psi \in \D (A \otimes A^\prime)$ is
a pure state satisfying $H(A) = H(A^\prime) = \log \dim (A) = \log \dim (A^\prime)$.
For a system $A$ of dimension $\dim (A)$ and any $\rho \in \D(A \otimes B \otimes C )$, we have the bounds
\begin{align}
	0 \leq H(A) \leq \log \dim (A), \\
	- H (A) \leq H(A | B) \leq H (A), \\
	0 \leq I(A; B) \leq 2 H (A), \\
	0 \leq I(A; B | C) \leq 2 H (A).
\end{align}
If $A$ or $B$ is a classical system, we get the tighter bounds
\begin{align}
	0  \leq H(A | B) , \\
	 I(A; B)  \leq  H (A), \\
	 I(A; B | C)  \leq  H (A).
\end{align}
The conditional mutual information satisfies a chain rule: for any $\rho \in \D(A \otimes B \otimes C \otimes D)$,
\begin{align}
	I (AB ; C | D) = I (A; C | D) + I(B; C | A D).
\end{align}
For product states $\rho^{A_1 B_1 C_1 A_2 B_2 C_2} = \rho_1^{A_1 B_1 C_1} \otimes \rho_2^{A_2 B_2 C_2}$, entropy is additive,
\begin{align}
	H(A_1 A_2) = H(A_1) + H(A_2),
\end{align}
and so there is no conditional mutual information between product system,
\begin{align}
	I(A_1 ; A_2 | B_1 B_2 ) = 0,
\end{align}
and conditioning on a product system is useless,
\begin{align}
	I(A_1 ; B_1 | C_1 A_2 ) = I(A_1 ; B_1 | C_1 ).
\end{align}
More generally,
\begin{align}
	I(A_1 A_2 ; B_1 B_2 | C_1 C_2 ) = I(A_1 ; B_1 | C_1 ) + I(A_2; B_2 | C_2).
\end{align}
Two important properties of the conditional mutual information are non-negativity, equivalent to strong subadditivity, and the data processing inequality.
%we consider an equivalent rewriting of strong subadditivity, which states that conditional mutual information is non-negative.
For any $\rho \in \D (A \otimes B \otimes C)$ and $\N \in \C (B, B^\prime)$, with $\sigma = \N (\rho)$,
\begin{align}
	I (A; B | C)_\rho & \geq 0, \\
	I (A; B | C)_\rho & \geq I (A; B^\prime | C)_\sigma.
\end{align}
%We will also make use of the following continuity property of conditional mutual information, following from the Alicki-Fannes inequality
%for conditional von Neumann entropy~\cite{AF}. F
%It is also continuous~\cite{AF04}: for any two states $\rho_1, \rho_2 \in \D (A \otimes B \otimes C)$ with $\|\rho_1 - \rho_2 \|_{ABC} \leq \epsilon$, it holds that
%\begin{align}
%	| H(A|B)_{\rho_1} - H(A|B)_{\rho_2}  | \leq 4 \epsilon \log \dim (A) + 2 H_2 (\epsilon), \\
%	| I (A : B | C)_{\rho_1} - I (A : B | C)_{\rho_2}  | \leq 8 \epsilon \log \dim (A) + 4 H_2 (\epsilon),
%\end{align}
%in which $H_2 (\epsilon)$ is the binary entropy function.
For classical systems, conditioning is equivalent to taking an average:
for any $\rho^{ABCX} = \sum_x p_X(x) \cdot \kb{x}{x}^X \otimes \rho_x^{ABC}$, for a classical system $X$ and
some appropriate $\rho_x \in \D (A \otimes B \otimes C)$,
\begin{align}
	H (A | B X )_\rho & = \sum_x p_X (x) \cdot H (A | B)_{\rho_x},\\
	I (A; B | C X )_\rho & = \sum_x p_X (x) \cdot I(A; B | C)_{\rho_x}.
\end{align}

%------------------------------------------------------------
\subsection{Quantum Communication Model}
\label{sec:qucomm}

The model for 
communication complexity that we consider is the following. For a given 
bipartite relation $T \subset X \times Y \times Z_A \times Z_B$
and input distribution $\mu$ on $X \times Y$, Alice and Bob are 
given input registers $A_{in}, B_{in}$ containing their classical input $x \in X, y \in Y$ at the outset of the protocol, respectively,
and they output registers $A_{out}, B_{out}$ containing their classical output $z_A \in Z_A, z_B \in Z_B$ at the end of the protocol, 
respectively, which should satisfy the relation $T$. We generally allow for some small
error $\epsilon$ in the output, which will be formalized below. 
In this distributional communication complexity setting, the input is a classical state $\rho = \sum_{x \in X, y \in Y} \mu (x, y) \cdot \kb{x}{x}^{A_{in}} \otimes \kb{y}{y}^{B_{in}}$, similarly for the output $\Pi (\rho) = \sum_{z_A \in Z_A, z_B \in Z_B} p_{Z_A Z_B} (z_A, z_B) \cdot \kb{z_A}{z_A}^{A_{out}} \otimes \kb{z_B}{z_B}^{B_{out}}$ of the protocol $\Pi$ implementing the relation, and the error parameter corresponds to the average probability of failure $\sum_{x, y} \mu (x, y) \cdot [(x, y, \Pi (x, y)) \not\in R ] \leq \epsilon$.

A $r$-round protocol $\Pi$ for implementing relation $T$ on input $\rho^{A_{in} B_{in}}$ is 
defined by a sequence of isometries $U_1, \cdots, U_{r + 1}$ along with a 
pure state $\psi \in \D (T_A \otimes T_B)$ shared between Alice and Bob, 
for arbitrary finite dimensional registers $T_A, T_B$. 
%We use interchangeably the terms message and round, and the same holds for the abbreviations $M$ and $r$.
For appropriate finite 
dimensional memory registers $A_1, A_3, \cdots A_{r - 1}, A^\prime$ held by Alice, $B_2, B_4, \cdots B_{r - 2}, B^\prime$ 
held by Bob, and communication registers $C_1, C_2, C_3, \cdots C_r$ 
exchanged by Alice and Bob, we have
$U_1 \in \U(A_{in} \otimes T_A, A_1 \otimes C_1), 
U_2 \in \U(B_{in} \otimes T_B \otimes C_1, B_2 \otimes C_2), 
U_3 \in \U(A_1 \otimes C_2, A_3 \otimes C_3), 
U_4 \in \U(B_2 \otimes C_3, B_4 \otimes C_4), \cdots , 
U_{r} \in \U(B_{r - 2} \otimes C_{r - 1}, B_{out} \otimes B^\prime \otimes C_{r}), 
U_{r  + 1} \in \U(A_{r - 1} \otimes C_r, A_{out} \otimes A^\prime)$.
We adopt the convention that, in the first round, $B_1 = B_0 = B_{in} \otimes T_B$, in even rounds $B_i = B_{i-1} $, and in odd rounds $A_i = A_{i-1} $. In this way, in round $i$, after application of $U_i$, Alice holds register $A_i$, Bob holds register $B_i$ and the communication register is $C_i$. We slightly abuse notation and also write $\Pi$ to denote the channel implemented by the protocol, i.e.
\begin{align}
\Pi (\rho) =
\Tr{A^\prime B^\prime }{U_{r  + 1} U_r \cdots U_2 U_1 (\rho \otimes \psi)}.
\end{align}
To formally define the error, we introduce a purification register $R$.
For a classical input $\rho^{A_{in} B_{in}} = \sum_{x \in X, y \in Y} \mu (x, y) \cdot \kb{x}{x}^{A_{in}} \otimes \kb{y}{y}^{B_{in}}$ like we consider here,
we can always take this purification to be of the form 
$\ket{\rho}^{A_{in} B_{in} R} = \sum_{x \in X, y \in Y} \sqrt{\mu (x, y)} \ket{x}^{A_{in}} \ket{y}^{B_{in}} \ket{xy}^{R_1} \ket{xy}^{R_2}$, for an appropriately chosen partition of $R$ into $R_1, R_2$. Note that if we trace out the the $R_2$ register, then we are left with a classical state such that $R_1$ contains a copy of the joint input.
Then we say that a protocol $\Pi$ for implementing relation $T$
on input $\rho^{A_{in} B_{in}}$, with purification $\rho^{A_{in} B_{in} R}$,
has average error $\epsilon \in [0, 1]$ if $P_e^\mu = $Pr$_{\mu, \Pi} [\Pi (\rho^{A_{in} B_{in} R_1}) \not\in T] \leq \epsilon$.
We denote the set of all such protocols as $ \T (T, \mu, \epsilon)$.
If we want to restrict this set to bounded round protocols with  $r$ rounds, we write $\T^r (T, \mu, \epsilon)$.
The worst case error of a protocol is $P_e^w = \max_\mu P_e^\mu$, in which it is sufficient to optimize over all atomic distributions $\mu$. We denote by $\T (T, \epsilon)$ the set of all protocols implementing relation $T$ with worst case error at most $\epsilon$, and by $\T^r (T, \epsilon)$ if we restrict this set to $r$-round protocols.
Let us formally 
define the different quantities that we work with.

\begin{definition}
For a protocol $\Pi$ as defined above, we define the \emph{quantum communication cost} of $\Pi$ as
\begin{align*}
	QCC (\Pi) = \sum_{i} \log \dim (C_i).
\end{align*}
\end{definition}

Note that we do not require that $\dim (C_i) = 2^k$ for some $k \in \mathbb{N}$, as is usually done.
This will not affect our definition on information cost and complexity, but
might affect the quantum communication complexity by at most a factor of two, without affecting the round complexity.
The corresponding notions of quantum communication complexity of a relation are:

\begin{definition}
For a relation $T \subset X \times Y \times Z_A \times Z_B$, an input distribution $\mu$ on $X \times Y$ and an
error parameter $\epsilon \in [0, 1]$, we define the $\epsilon$-error \emph{quantum communication complexity} of $T$ 
on input $\mu$ as
\begin{align*}
	QCC (T, \mu, \epsilon) = \min_{\Pi \in \T (T, \mu, \epsilon)} QCC (\Pi),
\end{align*}
and the worst-case $\epsilon$-error \emph{quantum communication complexity} of $T$ as
\begin{align*}
	QCC (T, \epsilon) = \min_{\Pi \in \T (T, \epsilon)} QCC (\Pi),
\end{align*}

\end{definition}

\begin{remark} 
For any $T, \mu, 0 \leq \epsilon_1 \leq \epsilon_2 \leq 1$, the following holds:
\begin{align*}
QCC(T, \mu, \epsilon_2) & \leq QCC(T, \mu, \epsilon_1), \\
QCC(T, \epsilon_2) & \leq QCC(T, \epsilon_1).
\end{align*}
\end{remark}

We have the following definitions for bounded round quantum communication complexity, and
a similar remark holds.

\begin{definition}
For a relation $T \subset X \times Y \times Z_A \times Z_B$, an input distribution $\mu$ on $X \times Y$, an
error parameter $\epsilon \in [0, 1]$
 and a bound $r \in \mathbb{N}$ on the number of rounds, 
we define the $r$-round, $\epsilon$-error \emph{quantum communication complexity} of $T$ 
on input $\mu$ as
\begin{align*}
	QCC^r (T, \mu, \epsilon) = \min_{\Pi \in \T^r (T, \mu, \epsilon)} QCC (\Pi),
\end{align*}
and $r$-round, worst-case $\epsilon$-error \emph{quantum communication complexity} of $T$ as
\begin{align*}
	QCC^r (T, \epsilon) = \min_{\Pi \in \T^r (T, \epsilon)} QCC (\Pi),
\end{align*}

\end{definition}

%------------------------------------------------------------
\subsection{Quantum Information Complexity}

We use the notion of quantum information complexity as defined in~\cite{Tou14}.
The register $R$ is the purification register, invariant throughout the protocol since we consider local isometric processing. Note that,
as noted before when considering a $R_1 R_2$ partition for $R$, for classical input distributions, the purification register can be thought of as containing a (quantum) copy of the classical input. The definition is however invariant under the choice of $R$ and corresponding purification.

\begin{definition}
For a protocol $\Pi$ and a state $\rho$
% on $X \times Y$, 
with purification held in system $R$,
we define the \emph{quantum information cost} of $\Pi$ on input $\rho$ as
\begin{align*}
	QIC (\Pi, \rho) = \sum_{i>0, odd} \frac{1}{2} I(C_i; R | B_{i }) + \sum_{i>0, even} \frac{1}{2} I(C_i; R | A_{i }).
\end{align*}
%in which we have labelled $B_0 = B_{in} \otimes T_B$.
%, and $B_N = B_{out} \otimes B^\prime$.
\end{definition}

\begin{definition}
For a relation $T \subset X \times Y \times Z_A \times Z_B$, an input distribution $\mu$ on $X \times Y$,
an error parameter $\epsilon \in [0, 1]$ and a number of round $r$,
we define the $\epsilon$-error \emph{quantum information complexity} of $T$ on input $\mu$ as
\begin{align*}
	QIC (T, \mu, \epsilon) = \inf_{\Pi \in \T (T, \mu, \epsilon)} QIC (\Pi, \mu),
\end{align*}
and the $r$-round, $\epsilon$-error \emph{quantum information complexity} of $T$ on input $\mu$ as
\begin{align*}
	QIC^r (T, \mu, \epsilon) = \inf_{\Pi \in \T^r (T, \mu, \epsilon)} QIC (\Pi, \mu),
\end{align*}
\end{definition}

The following properties of quantum information cost and complexity were proved in Ref.~\cite{Tou14}.

\begin{lemma}
\label{lem:qicvsqcc}
	For any protocol $\Pi$ and input distribution $\mu$, the following holds:
\begin{align*}
0 \leq QIC (\Pi, \mu) \leq QCC (\Pi).
\end{align*}
\end{lemma}

\begin{lemma}
For a relation $T \subset X \times Y \times Z_A \times Z_B$, an input distribution $\mu$ on $X \times Y$,
an error parameter $\epsilon \in [0, 1]$ and a number of round $r$,
the following holds:
\begin{align*}
0 \leq QIC (T, \mu, \epsilon) \leq QCC (T, \mu, \epsilon), \\
0 \leq QIC^r (T, \mu, \epsilon) \leq QCC^r (T, \mu, \epsilon).
\end{align*}
\end{lemma}

\begin{lemma}
\label{lem:add1}
	For any two protocols $\Pi^1$ and $ \Pi^2$ with $r_1$ and $ r_2$ rounds, respectively,
there exists a $r$-round protocol $\Pi_2$, satisfying $\Pi_2 = \Pi^1 \otimes \Pi^2, r = \max (r_1, r_2)$, such that the following holds
for any corresponding input states $\rho^1, \rho^2$:
\begin{align*}
QIC (\Pi_2, \rho^1 \otimes \rho^2) = QIC(\Pi^1, \rho^1) + QIC(\Pi^2, \rho^2).
\end{align*}
\end{lemma}

\begin{lemma}
\label{lem:add2}
	For any $r$-round protocol $\Pi_2$ and any input states $\rho^1 \in \D (A_{in}^1 \otimes B_{in}^1), \rho_2 \in \D ( A_{in}^2 \otimes B_{in}^2)$,
there exist $r$-round protocols $\Pi^1, \Pi^2$ satisfying $\Pi^1 (\cdot) = \Tra{A_{out}^2 B_{out}^2} \circ \Pi_2 (\cdot \otimes \rho^2)$ , $\Pi^2 (\cdot) = \Tra{A_{out}^1 B_{out}^1} \circ \Pi_2 (\rho^1 \otimes \cdot)$, and the following holds:
\begin{align*}
QIC(\Pi^1, \rho^1) + QIC(\Pi^2, \rho^2) = QIC (\Pi_2, \rho^1 \otimes \rho^2).
\end{align*}
\end{lemma}

\begin{lemma}
\label{lem:conv}
For any $p \in [0, 1]$, any two protocols $\Pi^1, \Pi^2$ with $r_1, r_2$ rounds, respectively, 
there exists a $r$-round protocol $\Pi$ satisfying
$\Pi = p\Pi^1 +(1-p)\Pi^2, r = \max (r_1, r_2)$, such that the following holds for any state $\rho$:
\begin{align*}
 QIC(\Pi, \rho) = p QIC(\Pi^1, \rho) + (1-p) QIC(\Pi^2, \rho).
\end{align*}
\end{lemma}

\begin{corollary}
\label{cor:conv}
For any $p \in [0, 1], T$ and $\epsilon, \epsilon_1, \epsilon_2 \in [0, 1]$ satisfying $ \epsilon = p\epsilon_1 + (1-p) \epsilon_2$, for any bound  $r = \max (r_1, r_2 ), r_1, r_2 \in \mathbb{N}$ on the number of rounds and for any input distribution $\mu$ on $X \times Y$, the following holds:
\begin{align*}
 QIC(T, \mu, \epsilon) & \leq p QIC(T, \mu, \epsilon_1) + (1-p) QIC(T, \mu, \epsilon_2), \\
 QIC^r (T, \mu, \epsilon) & \leq p QIC^{r_1} (T, \mu, \epsilon_1) + (1-p) QIC^{r_2} (T, \mu, \epsilon_2).
\end{align*}
\end{corollary}

\begin{lemma}
\label{concavity} 
Let $\nu$ be a distribution over input states $\rho$ and denote $\overline{\rho} := \E_{\rho \sim \nu} \rho$. Then for any protocol $\pi$,
\begin{align*} 
\mathbb{E}_{\rho \sim \nu} [ QIC(\pi, \rho) ] \leq QIC(\pi, \overline{\rho})
\end{align*} 
\end{lemma} 

\begin{lemma}
\label{lem:coding}
For any $r$-round protocol $\Pi$, any input distribution $\mu$ with copies of $x, y$ in $R_1$, and any $\epsilon \in (0, 2], \delta > 0$,
there exists a large enough $n_0(\Pi, \rho, \epsilon, \delta)$ such that for any $n \geq n_0$, there exists
a $r$-round protocol $\Pi_n$ satisfying
\begin{align*}
\| \Pi_n ((\rho^{A_{in} B_{in} R_1})^{\otimes n}) - \Pi^{\otimes n} ((\rho^{A_{in} B_{in} R_1})^{\otimes n}) \|_{(A_{out} B_{out} R_1)^{\otimes n}} \leq \epsilon, \\
\frac{1}{n} QCC (\Pi_n) \leq QIC(\Pi, \rho) + \delta.
\end{align*}
\end{lemma}

\subsection{Generalized Discrepancy Method}

Generalized discrepancy method, also known as smooth discrepancy method, is one of the strongest methods for proving lower bounds for quantum communication. 

\begin{definition} Let $f : \mathcal{X} \times \mathcal{Y} \rightarrow \{0,1\}$ be a boolean function. The $\delta$-generalized discrepancy bound of $f$, denoted by $GDM_{\delta}(f)$, is defined as: 
\begin{align*}
&GDM_{\delta}(f) = \max\{\text{$GDM_{\delta}^{\mu}(f)$: $\mu$ a distribution over $\mathcal{X} \times \mathcal{Y}$}\} \\
&GDM_{\delta}^{\mu}(f) = \max\{\text{$\log \left( \frac{1}{\text{disc}^{\mu}(g)} \right)$, $g : \mathcal{X} \times \mathcal{Y} \rightarrow \{0,1\}$, $\Pr_{(x,y) \sim \mu}[f(x,y) \neq g(x,y)] \le \delta$}\} \\
&\text{disc}^{\mu}(g) = \max \left\{ \text{$|\sum_{(x,y) \in R} (-1)^{g(x,y)} \cdot \mu(x,y)|$ : $R \in \mathcal{R}$} \right\}
\end{align*}
\end{definition}

\noindent Here $\mathcal{R}$ is the set of combinatorial rectangles $\mathcal{A} \times \mathcal{B}$, $\mathcal{A} \subseteq \mathcal{X}, \mathcal{B} \subseteq \mathcal{Y}$. We state two results on the generalized discrepancy method, both due to Sherstov~\cite{patternmatrix, sdp_sherstov}, which we will use to lower bound the quantum information complexity of disjointness. The first is a threshold direct product result that will be useful to prove that the generalized discrepancy method is a lower bound on the quantum information complexity of boolean functions, and the second is a lower bound on the generalized discrepancy for the disjointness function.

\begin{theorem}[\cite{sdp_sherstov}] {\label{thm:sdp_sherstov}} Let $\epsilon_{\text{sh}} > 0$ be a small enough absolute constant. Then for any boolean function $f$ , the following communication problem requires $\Omega(n GDM_{1/5}(f))$ qubits of communication $($with arbitrary entanglement$)$: Solving with probability $2^{- \epsilon_{\text{sh}} n}$, at least $(1-\epsilon_{\text{sh}}) n$ among $n$ instances of $f$. 
\end{theorem}

The disjointness function is defined as follows: for $x,y \in \{0,1\}^n \times \{0,1\}^n$, $DISJ_n(x,y) = 1$ if for all $i \in [n]$, $x_i \wedge y_i = 0$, and $0$ otherwise. We will need the following theorem. 

\begin{theorem}[\cite{patternmatrix}] {\label{thm:disj_sherstov}} $GDM_{1/5}(DISJ_n) \ge \Omega (\sqrt{n})$
\end{theorem}

\section{Properties of Quantum Information Complexity}

In this section, we prove general results about quantum information complexity that we use to obtain the main results.
These may be of independent interest.

\subsection{Prior-free Quantum Information Complexity}

%\begin{definition}
%For a protocol $\Pi$,
%we define the \emph{quantum communication cost} of $\Pi$ as
%\begin{align*}
%	QCC (\Pi) =  \sum_{i}  \log \dim (C_i ).
%\end{align*}
%\end{definition}

%\begin{definition}
%For a protocol $\Pi$ and an input distribution $\mu$ on $X \times Y$, 
%we define the \emph{quantum information cost} of $\Pi$ on input $\mu$ as
%\begin{align*}
%	QIC (\Pi, \mu) = \sum_{i>0, odd} \frac{1}{2} I(C_i; R | B_{i - 1} ) + \sum_{i>0, even} \frac{1}{2} I(C_i; R | A_{i - 1}),
%\end{align*}
%in which we have labelled $B_0 = B_{in} \otimes T_B$.
%%, and $B_N = B_{out} \otimes B^\prime$.
%\end{definition}

%We say that a protocol $\Pi$ for implementing relation $T$
%on input $\rho^{A_{in} B_{in}}$ distributed according to $\mu$, with purification $\rho^{A_{in} B_{in} R}$,
%has error $\epsilon \in [0, 1]$ if $P_{e, \mu} = $Pr$_{\mu} [\Pi (\rho^{A_{in} B_{in} R_1}) \not\in T] \leq \epsilon$.
%We denote the set of all such protocols as $ \T (T, \mu, \epsilon)$. If we restrict this set to protocols with $M$ messages, with denote
%it as $\T^M (T, \mu, \epsilon)$.
%Similarly, we denote by $ \T (T, \epsilon)$ the set of all protocols $\Pi$ implementing 
%relation $T$ with error $\epsilon$ on all inputs,
%i.e.~for all $x, y$ we have $P_{e, x, y} = $Pr$ [\Pi (x, y) \not\in T] \leq \epsilon$.
%The restriction to $M$-message protocols is denoted $\T^M (T, \epsilon)$.
%We denote by $\D_{XY}$ the set of all distributions $\mu$ on input space $X \times Y$.

We want to define a sensible notion of quantum information complexity for classical tasks.
Like in the classical setting~\cite{B12}, there are two sensible orderings for the optimization over inputs and protocols.
We provide the two corresponding definitions and then investigate the link between them.
We denote by $\D_{XY}$ the set of all distributions $\mu$ on input space $X \times Y$.

\begin{definition}

The \emph{max-distributional quantum information complexity} of a relation $T$ with error $\epsilon \in [0, 1]$ is
\begin{align*}
	QIC_D (T, \epsilon) = \max_{\mu \in \D_{XY} } QIC (T, \mu, \epsilon).
\end{align*}
When restricting to $r$-round protocols, it is 
\begin{align*}
	QIC_D^r (T, \epsilon) = \max_{\mu \in \D_{XY} } QIC^r (T, \mu, \epsilon).
\end{align*}

\end{definition}

\begin{definition}

The \emph{quantum information complexity} of a relation $T$ with error $\epsilon \in [0, 1]$ is
\begin{align*}
	QIC (T, \epsilon) = \inf_{\Pi \in \T (T, \epsilon)} \max_{\mu \in \D_{XY} } QIC (\Pi, \mu).
\end{align*}
When restricting to $r$-round protocols, it is 
\begin{align*}
	QIC^r (T, \epsilon) = \inf_{\Pi \in \T^r (T, \epsilon)} \max_{\mu \in \D_{XY} } QIC (\Pi, \mu).
\end{align*}

\end{definition}

\begin{lemma}[Information lower bounds communication]
\label{lem:qicvsqccpf}

For any relation $T$, error parameter $\epsilon \in [0, 1]$, and number of rounds $r \in \mathbb{N}$, the following holds:
\begin{align*}
	QIC^r (T, \epsilon) & \leq QCC^r (T, \epsilon), \\
	QIC (T, \epsilon) & \leq QCC (T, \epsilon).
\end{align*}

\end{lemma}

\begin{proof}
%	Given $\delta > 0$, l
Let $\Pi$ be a protocol computing $T$ correctly except with probability $\epsilon$ on 
all input  and satisfying $QCC (\Pi) = QCC (T, \epsilon)$. We get the result by 
noting that $QIC (T, \epsilon) \leq \max_\mu QIC (\Pi, \mu) \leq QCC (\Pi)$.
\end{proof}

Clearly, $QIC_D (T, \epsilon) \leq QIC (T, \epsilon)$, and $QIC_D^r (T, \epsilon) \leq QIC^r (T, \epsilon)$. 
%At least in the bounded round case, we may almost 
We prove that we can almost reverse the quantifiers.
The proof idea follows the lines of the proof of Theorem 3.5 in Ref.~\cite{B12}, but special care
must be taken for quantum protocols.
The idea we use is to take an $\epsilon$-net over $\D_{XY}$, and then take a $\delta$-optimal protocol for each distribution in the net.
To extend this result to the unbounded round quantum setting, we adapt a compactness argument from Ref.~\cite{BGPW12}, itself adapted from Ref.~\cite{Ter72}.
The following results will be used.

\begin{lemma}[Continuity in average error]
\label{cor:cty_comp_input}
Quantum information complexity is continuous in the error. This holds uniformly in the input. That is, 
for all $T, r$ and $ \epsilon, \delta > 0$, there exists $\epsilon^\prime \in (0,  \epsilon)$ such 
that for all $\epsilon^{\prime \prime} \in (\epsilon^\prime,  \epsilon)$ and for all $\mu$,
\begin{align*}
	|QIC (T, \mu, \epsilon - \epsilon^{\prime \prime}) - QIC (T, \mu, \epsilon)| & \leq \delta, \\
	|QIC^r (T, \mu, \epsilon - \epsilon^{\prime \prime}) - QIC^r (T, \mu, \epsilon)| & \leq \delta. \\
\end{align*}

\end{lemma}

\begin{proof}
	Note that we can drop the absolute values and also work at $\epsilon^{\prime}$ since quantum information complexity is non-increasing in the error,
i.e.~$QIC (T, \mu, \epsilon) \leq QIC (T, \mu, \epsilon - \epsilon^{\prime \prime}) \leq QIC (T, \mu, \epsilon - \epsilon^{\prime})$.
Let $0 < p < \frac{1}{2}$ and use Corollary~\ref{cor:conv} with $\epsilon_1 = 0, \epsilon_2 = \epsilon, \epsilon^\prime = p \epsilon$ for the current $\epsilon$.
We get
\begin{align*}
	QIC (T, \mu, \epsilon - \epsilon^\prime) & \leq p QIC (T, \mu, 0)  + (1 - p) QIC (T, \mu, \epsilon )  \\
		& \leq p QCC (T, 0) + QIC (T, \mu, \epsilon).
\end{align*}
Rearranging terms, we get
\begin{align*}
	|QIC (T, \mu, \epsilon - \epsilon^\prime) - QIC (T, \mu, \epsilon)| & \leq \frac{\epsilon^\prime}{\epsilon} QCC (T, 0).
\end{align*}
This bound is independent of $\mu$, and goes to zero as $p$ and $\epsilon^\prime$ do, so the result follows.
The bounded round result is proved in the same way, obtaining $QCC^r (T, 0)$ in the final bound instead.
\end{proof}

\begin{lemma}[Convexity in error]
For any $p \in [0, 1], T$ and $\epsilon, \epsilon_1, \epsilon_2 \in [0, 1]$ satisfying $ \epsilon = p\epsilon_1 + (1-p) \epsilon_2$ and for any bound $r = \max (r_1, r_2 ), r_1, r_2 \in \mathbb{N}$ on the number of rounds, the following holds:
\begin{align*}
 QIC(T, \epsilon) & \leq p QIC(T, \epsilon_1) + (1-p) QIC(T, \epsilon_2), \\
 QIC^r (T, \epsilon) & \leq p QIC^{r_1} (T, \epsilon_1) + (1-p) QIC^{r_2} (T, \epsilon_2).
\end{align*}
\end{lemma}

\begin{proof}
The proof is similar to the one for the analogous result with fixed input.
Given $\delta > 0$, let $\Pi^1$ and $\Pi^2$ be protocols satisfying, for all $\mu$,
 for $i \in \{ 1, 2\}, \Pi^i \in \T (T, \epsilon_i), QIC (\Pi^i, \mu) 
\leq QIC (T, \epsilon_i) + \delta$, and take the corresponding protocol $\Pi$ of Lemma \ref{lem:conv}.
First, it holds that protocol $\Pi$ successfully accomplish its task, i.e.~it implements task $T$ on all inputs with error bounded
by $\epsilon = p \epsilon_1 + (1-p) \epsilon_2$.
We must now verify that the quantum information cost satisfies the convexity property:
\begin{align*}
 QIC(T, \epsilon) & \leq \max_\mu QIC (\Pi, \mu)  \\
		&= \max_\mu \big( p QIC (\Pi^1, \mu) +  (1-p) QIC (\Pi^2, \mu) \big) \\
		&\leq p \max_\mu  QIC (\Pi^1, \mu) +  (1-p) \max_\mu QIC (\Pi^2, \mu) \\
		&\leq p QIC(T, \epsilon_1) + (1-p) QIC(T, \epsilon_2) + 2 \delta.
\end{align*}
Keeping track of rounds, we get the bounded round result.
\end{proof}

\begin{corollary}[Continuity in error]
Quantum information complexity is continuous in the error. That is, 
for all $T, r$ and $ \epsilon, \delta > 0$, there exists $\epsilon^\prime \in (0,  \epsilon)$ such that for all $\epsilon^{\prime \prime} \in (\epsilon^\prime, \epsilon)$
\begin{align*}
	|QIC (T, \epsilon - \epsilon^{\prime \prime}) - QIC (T, \epsilon)| & \leq \delta, \\
	|QIC^r (T, \epsilon - \epsilon^{\prime \prime}) - QIC^r (T, \epsilon)| & \leq \delta. \\
\end{align*}

\end{corollary}

\begin{lemma} [Quasi-convexity in input]
For any $p \in [0, 1]$, define $\rho = p\rho_1 + (1-p) \rho_2$ for any two input states $\rho_1, \rho_2$.
Then the following holds for any $r$-round protocol $\Pi$:
\begin{align*}
 	QIC(\Pi, \rho) & \geq p QIC(\Pi, \rho_1) + (1-p) QIC(\Pi, \rho_2) \\
	QIC(\Pi, \rho) & \leq p QIC(\Pi, \rho_1) + (1-p) QIC(\Pi, \rho_2) +  r H(p).
\end{align*}
\end{lemma}

\begin{proof}
	The first inequality is Lemma~\ref{concavity}, and the second is obtained by keeping track of the remainder terms discarded in its proof.
Let $R$ be a register holding a purification of $\rho_1$ and $\rho_2$, then we can purify $\rho$ with two copies $S_1, S_2$ of a selector reference register, such that $\ket{\rho}^{A_{in} B_{in} R S_1 S_2} = \sqrt{p} \ket{\rho_1}^{A_{in} B_{in} R} \ket{1}^{S_1} \ket{1}^{S_2} + \sqrt{1-p} \ket{\rho_2}^{A_{in} B_{in} R} \ket{2}^{S_1} \ket{2}^{S_2}$. We can then expand each term as 
\begin{align*}
	I (C_i ; R S_1 S_2 | B_i)_\rho = I (C_i ; S_1 | B_i)_\rho + I (C_i ; R | B_i S_1)_\rho + I (C_i ; S_2 | B_i R S_1)_\rho,
\end{align*}
 and similarly for terms conditioning on Alice's systems $A_i$. The result follows by summing over all rounds since
\begin{align*}
	I (C_i ; R | B_i S_1)_\rho = p I (C_i ; R | B_i)_{\rho_1} + (1 - p) \cdot I (C_i ; R | B_i)_{\rho_2},
\end{align*}
 and then $H(S) = H(p)$ upper bounds the two remainder terms in each of the $r$ rounds.
\end{proof}

\begin{lemma}[Continuity in input]
\label{continuity:qic}
%\label{cor:cty_cost_input}
Quantum information cost for $r$-round protocols is uniformly continuous in the input distribution. 
This holds uniformly over all $r$-round protocols over input $X \times Y$. That is, for all $r, |X|, |Y|$,
 and $\epsilon > 0$, there exists $\delta > 0$ such 
that for all $\mu_1$ and $\mu_2$ that are $\delta$-close and all $r$-round protocols $\Pi$, 
\begin{align*}
	|QIC (\Pi, \mu_1) - QIC (\Pi, \mu_2)| \leq \epsilon.
\end{align*}

\end{lemma}

\begin{proof}

Let $\delta > 0$ and fix $\mu_1$ and $\mu_2$ that are $\delta$-close.
We can then write, for some common part $\mu_0$ and remainder parts $\mu_1^\prime, \mu_2^\prime$,
\begin{align*}
	&\mu_1 = (1 - \delta) \mu_0 + \delta \mu_1^\prime, \\
	&\mu_2 = (1 - \delta) \mu_0 + \delta \mu_2^\prime, \\
	&\mu_0(x,y) = \frac{\min(\mu_1(x,y), \mu_2(x,y))}{\sum_{x',y'} \min(\mu_1(x',y'), \mu_2(x',y'))}.
\end{align*}

Using the bounds in the lemma above once on each of $\mu_1$ and $\mu_2$, we get
\begin{align*}
	QIC (\Pi, \mu_1) & \leq (1- \delta) QIC (\Pi, \mu_0) + \delta QIC (\Pi, \mu_1^\prime) +  r H(\delta) \\ 
		&  \leq (1- \delta) QIC (\Pi, \mu_0) + \delta QIC (\Pi, \mu_2^\prime) + \delta QIC (\Pi, \mu_1^\prime) +  r H(\delta) \\ 
		& \leq  QIC (\Pi, \mu_2) +  \delta \cdot  r  (\log |X| + \log |Y|)  +  r H(\delta).
\end{align*}
Similarly, we get a bound on $QIC (\Pi, \mu_2)$ in terms of $QIC (\Pi, \mu_1)$, so the following holds: 
\begin{align*}
	| QIC (\Pi, \mu_1) -QIC(\Pi, \mu_2) | \leq \delta \cdot  r  (\log |X| + \log |Y|) +  r H(\delta).
\end{align*}
This bound is independent of $\mu_1, \mu_2$, depends on $\Pi$ only through $r$  and $|X|, |Y|$, and goes to zero as $\delta$ does, so the result follows.
\end{proof}

%\begin{proof}

%In the above proof, the final bound makes no mention of the protocol except for the number $M$ of messages (and  the input space).

%\end{proof}

\begin{corollary}
\label{onezero} %% Lemma on conditioning 
Suppose we have a $r$-round protocol $\Pi$ for $AND$. 
Then,
\begin{equation}
QIC(\Pi, \mu) \leq QIC(\Pi, \mu_0)  + O( r H(w))
\end{equation}
where $w = \mu(1,1) \leq 1/2$, $\mu_0(1,1) = 0$, and $\mu_0(x_i,y_i) = \frac{1}{1-w} \mu(x_i,y_i)$ otherwise.
\end{corollary}

\begin{proof}
This just follows from the proof of lemma \ref{continuity:qic}, since the input size is constant. 
\end{proof}

\begin{theorem}
\label{thm:minimax}
For a relation $T \subset X \times Y \times Z_A \times Z_B$,
an error parameter $\epsilon \in (0, 1)$, a number of rounds $r$ and each value $ \alpha \in (0, 1)$, 
\begin{align*}
	QIC^r (T, \frac{\epsilon}{\alpha}) \leq  \frac{QIC_D^r (T, \epsilon)}{1 - \alpha}.
\end{align*}

\end{theorem}

\begin{proof}

Fix $T, r, \epsilon, \alpha$ and denote $I = QIC_D^r (T, \epsilon)$. 
For any $\delta_1 \in (0, 1)$, we want to prove the existence of a protocol $\Pi \in \T^r (T, \frac{\epsilon}{\alpha} \cdot (1 + 2 \delta_1))$ 
satisfying $QIC (\Pi, \mu) \leq  \frac{I \cdot  (1 + 2 \delta_1)}{1 - \alpha} $ for all $\mu \in \D_{XY}$.
This shows that $QIC^r (T, \frac{\epsilon}{\alpha} \cdot (1 + 2 \delta_1)) \leq \frac{I}{1 -\alpha} \cdot(1 + 2 \delta_1)$,
and then by continuity of quantum information complexity in the error, we get the result by taking $\delta_1$ to $0$.
The proof follows along the lines of the one for the analogous result for 
classical information complexity [Bra12], using a minimax argument. 
We take extra care to account for the continuum of quantum protocols, the round-by-round definition of quantum information cost,
and the fact that we do not have a bound on the size of the entanglement.
%(and number of rounds).
Let $\delta_2 \in (0, \epsilon \delta_1)$ satisfy the following two properties for all $\mu_1, \mu_2$ that are $\delta_2$-close, and for all $r$-round protocols $\Pi$: 
\begin{align}
\label{eq:ctnycost}
	| QIC (\Pi, \mu_1) -  QIC (\Pi, \mu_2) | & \leq I \cdot \frac{\delta_1}{10},\\
\label{eq:ctnycomp}
	|QIC^r (T, \mu_1, \epsilon - \delta_2) - QIC^r (T, \mu_1, \epsilon)| & \leq I \cdot \frac{\delta_1}{10}.
%	\delta_2 & \leq \epsilon \delta_1.
\end{align}
The first inequality is possible by Lemma~\ref{continuity:qic}, i.e.~by the uniform continuity of quantum information cost in the input, uniformly over all $r$-rounds protocols, and the second
 is possible by Lemma~\ref{cor:cty_comp_input}, i.e.~the continuity of quantum information complexity in the error, uniformly over all inputs.
Fix a finite $\delta_2$-net for $\D_{XY}$, 
that we denote $N_{XY}$. For each $\mu \in N_{XY}$, fix a protocol $\Pi_\mu \in \T^r (T, \mu, \epsilon - \delta_2)$ 
such that $QIC(\Pi_\mu, \mu) \leq QIC^r (T, \mu, \epsilon - \delta_2)\cdot (1 + \frac{\delta_1}{10})$ 
and denote the set of all such protocols $P_{N}$. 
 We then have 
$|P_{N}| = |N_{XY}| < \infty$, and we get using (\ref{eq:ctnycomp}) that 
\begin{align}
	QIC(\Pi_\mu, \mu) & \leq QIC^r  (T, \mu, \epsilon - \delta_2)\cdot (1 + \frac{\delta_1}{10})  \nonumber \\
		& \leq \big( QIC^r (T, \mu, \epsilon) + I \cdot \frac{\delta_1}{10} \big) (1 + \frac{\delta_1}{10}) \nonumber \\
		& \leq I (1 + \frac{\delta_1}{10})^2 \nonumber \\
\label{eq:mumu}
		& \leq I (1 + \frac{\delta_1}{2}).
\end{align}
We define the following two-player 
zero-sum game over these two sets. Player $A$ comes up with a 
quantum protocol $\Pi \in P_N$. Player $B$ comes up with a distribution 
$\mu \in N_{XY}$. Player $B$'s payoff is given by 
\begin{align*}
	P_B (\Pi, \mu) = (1 - \alpha) \cdot \frac{QIC (\Pi, \mu)}{I} + 
		\alpha \cdot \frac{Pr_{\mu} [\Pi \not \in T]}{\epsilon},
\end{align*}
and then player $A$'s is given by $P_A (\Pi, \mu) = - P_B (\Pi, \mu)$. We first show the following.
\begin{claim}
	The value of the game for player $B$ is bounded by $1+\delta_1$.
%: $V_B (\Pi, \mu) \leq 1 + \delta_1$.
\end{claim}

\begin{proof}
	Let $\nu_B$ be a probability distribution over $N_{XY}$ representing a mixed strategy for player $B$.
To prove the claim, it suffices to show  that there is a protocol $\Pi \in P_N$ 
such that $\mathbb{E}_{\nu_B} [P_B (\Pi, \mu)] < 1 + \delta_1$. 
Let $\bar{\mu}$ be the distribution corresponding to averaging over $\nu_B$, that is
\begin{align*}
	\bar{\mu} (x, y) = \mathbb{E}_{\nu_B} \mu(x, y).
\end{align*}
Let $\mu^\prime \in N_{XY}$ be a distribution that is $\delta_2$-close 
to $\bar{\mu}$, and $\Pi^\prime \in P_N$ the corresponding protocol. 
We will show that $\Pi^\prime$ is also good for $\bar{\mu}$. We first have
\begin{align*}
	Pr_{\bar{\mu}} [\Pi^\prime \not \in T] &\leq Pr_{\mu^\prime} [\Pi^\prime \not \in T] + \delta_2 \\
			&\leq \epsilon - \delta_2 + \delta_2 \\
			&= \epsilon,
\end{align*}
in which the first inequality follows from the fact that $\bar{\mu}$ 
and $\mu^\prime$ are $\delta_2$-close and the second inequality 
from the fact that $\Pi^\prime \in P_N$ is the protocol corresponding to $\mu^\prime \in N_{XY}$, i.e.~$\Pi^\prime \in \T^r (T, \mu^\prime, \epsilon - \delta_2)$. We also have
\begin{align*}
	QIC (\Pi^\prime, \bar{\mu}) & \leq QIC (\Pi^\prime, \mu^\prime) + I \cdot \frac{\delta_1}{2} \\
			& \leq I \cdot (1 + \delta_1),
\end{align*}
in which the first inequality follows from (\ref{eq:ctnycost}) and the second from the fact that $\Pi^\prime \in P_N$ 
is the protocol corresponding to $\mu^\prime \in N_{XY}$ along with (\ref{eq:mumu}).
We obtain
\begin{align*}
	\mathbb{E}_{\nu_B} [P_B (\Pi^\prime, \mu)] & = \mathbb{E}_{\nu_B} 
		\big[ (1 - \alpha) \cdot \frac{QIC (\Pi^\prime, \mu)}{I} + \alpha 
			\cdot \frac{Pr_{\mu} [\Pi^\prime \not \in T]}{\epsilon} \big] \\
		& = (1 - \alpha) \cdot \mathbb{E}_{\nu_B} \big[ \frac{QIC (\Pi^\prime, \mu)}{I} \big] + \alpha \cdot
			 \frac{Pr_{\bar{\mu}} [ \Pi^\prime \not \in  T]}{\epsilon} \\
	& \leq (1 - \alpha) \cdot  \big[ \frac{QIC (\Pi^\prime, \bar{\mu})}{I} \big] + 
			\alpha \cdot \frac{Pr_{\bar{\mu}} [ \Pi^\prime \not \in  T]}{\epsilon} \\
	& < (1 - \alpha) \cdot  (1 + \delta_1) + \alpha \\
	& < 1 + \delta_1,
\end{align*}
in which the first equality is by definition, the second by linearity of expectation, 
the first inequality is by Lemma~\ref{concavity}, i.e.~concavity of quantum information cost in the input state, 
and the second inequality is by the above results about $\Pi^\prime$.
This concludes the proof of the claim.
\end{proof}

By the minimax theorem for zero-sum games, the above claim implies that there exists a probability distribution $\nu_A$ over $P_N$ representing  a  mixed strategy for player $A$ and such that the value of the game for player $B$ is at most $1+\delta_1$. That is, for all $\mu \in N_{XY}$,
\begin{align*}
	\mathbb{E}_{\nu_A} (P_B (\Pi, \mu)) < 1 +\delta_1.
\end{align*}

Let $\bar{\Pi} = \mathbb{E}_{\nu_A} (\Pi)$ be the $r$-round protocol obtained by publicly averaging over $\nu_A$,
as per Lemma~\ref{lem:conv}.
%; this can be done with pre-shared entanglement serving as shared randomness, and then by running all protocols in parallel, with all but one running on a dummy input, for example as in the proof of Lemma~\ref{lem:conv}.
This is the protocol we are looking for. The following claim holds.

\begin{claim}
	For all $\mu \in \D_{XY}$, $(1 - \alpha) \cdot \frac{QIC (\bar{\Pi}, \mu)}{I} + \alpha 
			\cdot \frac{Pr_{\mu} [\bar{\Pi} \not \in T]}{\epsilon} < 1 + 2 \delta_1$.
\end{claim}

\begin{proof}
Fix any $\mu \in \D_{XY}$, and let $\mu^\prime \in N_{XY}$ be a distribution that is $\delta_2$-close to $\mu$.
Then 
we obtain
\begin{align*}
	(1 - \alpha) \cdot \frac{QIC (\bar{\Pi}, \mu)}{I} + \alpha 
			\cdot \frac{Pr_{\mu} [\bar{\Pi} \not \in T]}{\epsilon} 
		& \leq (1 - \alpha) \cdot  \frac{QIC (\bar{\Pi}, \mu^\prime) + I \delta_1}{I} + \alpha 
			\cdot \frac{Pr_{\mu^\prime} [\bar{\Pi} \not \in T] + \delta_2}{\epsilon} \\
		& = (1 - \alpha) \cdot  \frac{QIC (\bar{\Pi}, \mu^\prime) }{I} + \alpha \cdot
			 \mathbb{E}_{\nu_A} \frac{Pr_{\mu^\prime} [ \Pi \not \in  T] }{\epsilon}  \\
		&\quad\quad\quad + (1 - \alpha ) \cdot  \delta_1 + \alpha \cdot \frac{\delta_2}{\epsilon}\\
	& \leq (1 - \alpha) \cdot  \mathbb{E}_{\nu_A} \big[ \frac{QIC (\Pi, \mu^\prime)}{I}  \big] + 
			\alpha \cdot \mathbb{E}_{\nu_A} \big[ \frac{Pr_{\mu^\prime} [ \Pi \not \in  T]}{\epsilon}  \big] + \delta_1 \\
	& = \mathbb{E}_{\nu_A} [P_B (\Pi, \mu^\prime)]  + \delta_1 \\
	& < 1 + 2 \delta_1,
\end{align*}
in which the first inequality follows from (\ref{eq:ctnycost}) and the fact that $\mu, \mu^\prime$ are $\delta_2$-close, 
the first equality is because we take expectation over a probability, 
the second inequality is because $\delta_2 \leq \epsilon \cdot \delta_1$ and by Lemma~\ref{lem:conv}, i.e.~by the convexity of quantum information cost in the protocol, the second equality is by linearity of expectation and the definition of $P_B (\Pi, \mu^\prime)$, and the last inequality is because $\nu_A$ represents the mixed strategy obtained by the minimax theorem.
Since this holds for all $\mu \in \D_{XY}$, this conclude the proof of the claim.
\end{proof}

To conclude the proof of the theorem, we first note that the above claim implies that for all $\mu \in\D_{XY}$,
\begin{align*}
	QIC (\bar{\Pi}, \mu) \leq \frac{I}{1-\alpha}(1 + 2 \delta_1),
\end{align*}
so $\bar{\Pi}$ satisfies the quantum information cost property we are looking for. Is left to verify that it also has low error on all inputs. The above claim also implies that for all $\mu$,
\begin{align*}
	Pr_{\mu} [\bar{\Pi} \not \in T] \leq \frac{\epsilon}{\alpha} \cdot (1 + 2 \delta_1).
\end{align*}
Letting $\mu$ run over all atomic distributions, we get the desired error property, and so
\begin{align*}
	QIC^r (T, \frac{\epsilon}{\alpha} \cdot (1 + 2 \delta_1)) \leq \frac{I}{1-\alpha}(1 + 2 \delta_1),
\end{align*}
as desired.
\end{proof}

\begin{theorem} \label{minimax_general}
	For a relation $T \subset X \times Y \times Z_A \times Z_B$,
an error parameter $\epsilon \in (0, 1)$ and each value $ \alpha \in (0, 1)$, 
\begin{align*}
	QIC (T, \frac{\epsilon}{\alpha}) \leq  \frac{QIC_D (T, \epsilon)}{1 - \alpha}.
\end{align*}

\end{theorem}

\begin{proof}
	Let $I = QIC_D (T, \epsilon)$, and denote by $P_e^\mu (\Pi)$ the average error of $\Pi$ for computing $T$ on $\mu$, and by $P_T$ the set of all protocols over the same input and output spaces as $T$. Then for any $\Pi$, $P_e^\mu (\Pi)$ is continuous in $\mu$ by properties of the statistical distance. Given $\delta > 0$, define
\begin{align*}
	A (\Pi) = \{ \mu \in \D_{XY} :  QIC (\Pi, \mu) \geq I + 2 \cdot \delta \ \text{or} \ P_e^\mu (\Pi) \geq \epsilon +  \delta \}.
\end{align*}
By continuity of $QIC (\Pi, \mu)$ and $P_e^\mu (\Pi)$ in $\mu$, these sets are closed for all $\Pi \in P_T$.
Then, by definition of $I$, for all $\mu$ there exists $\Pi_\mu \in \T (T, \mu, \epsilon)$ such that $QIC (\Pi_\mu, \mu) \leq I + \delta$, and so $\cap_{\Pi \in P_T} A(\Pi) = \emptyset$. Since $\D_{XY}$ is compact and the sets $A (\Pi)$ are closed, we get that there exists a finite set $Q \subset P_T$ such that $\cap_{\Pi \in Q} A (\Pi) = \emptyset$. We get that for all $\mu$, there exists $\Pi_\mu \in Q$ such that $QIC (\Pi_\mu, \mu) < I + 2 \delta$ and $P_e^\mu (\Pi_\mu) < \epsilon + \delta$.
Let $r_M = \max \{ r:$ there is $\Pi \in Q$ with $r$ rounds $ \}$, then 
\begin{align*}
	I + 2 \delta & \geq \max_\mu \min_{\Pi \in Q \cap \T (T, \mu, \epsilon + \delta)} QIC (\Pi, \mu) \\
		& \geq QIC_D^{r_M} (T, \epsilon +  \delta) \\
		& \geq  (1 - \alpha) \cdot QIC^{r_M} (T, \frac{\epsilon }{\alpha} +  \frac{\delta}{\alpha}) \\
		& \geq  (1 - \alpha) \cdot QIC (T, \frac{\epsilon }{\alpha} +  \frac{\delta}{\alpha}).
%		& \geq (1 - \alpha) \cdot QIC^{r_M} (T, \frac{\epsilon}{\alpha} ) - (1 - \alpha) \epsilon^\prime_{ \delta/\alpha} \\
%		& \geq (1 - \alpha) \cdot QIC (T, \frac{\epsilon}{\alpha} ) - (1 - \alpha) \epsilon^\prime_{ \delta/\alpha},
\end{align*}
%in which $\epsilon^\prime_{ \delta/\alpha}$ is given by the continuity of $QIC^{r_M}$, and goes to zero as $\delta$ does. 
The result
follows by continuity of $QIC$ and by taking $\delta$ to zero.
\end{proof}

\subsection{Subadditivity}

\begin{lemma}
\label{lem:subadd}
For any two protocols $\Pi^1, \Pi^2$ with $r_1, r_2$ rounds, respectively,
there exists a $r$-round protocol $\Pi_2$, satisfying $\Pi_2 = \Pi^1 \otimes \Pi^2, r = \max (r_1, r_2)$, such that the following holds
for any joint input state $\rho_{12} \in \mathcal{D} (A_{in}^1 \otimes B_{in}^1 \otimes A_{in}^2 \otimes B_{in}^2)$:
\begin{align*}
QIC (\Pi_2, \rho_{12}) \leq QIC(\Pi^1, \rho_1) + QIC(\Pi^2, \rho_2),
\end{align*}
with $\rho_1 = \Tr{A_{in}^2 B_{in}^2}{\rho_{12}}$ and $\rho_2 = \Tr{A_{in}^1 B_{in}^1}{\rho_{12}}$.
\end{lemma}

\begin{proof}
Given protocols $\Pi^1$ and $\Pi^2$, we assume without loss of generality that $r_1 \geq r_2$, and we define the protocol $\Pi_2$ in the following way.

%\begin{Protocol}
\begin{enumerate}
\item Run protocols $\Pi^1, \Pi^2$ in parallel for $r_2$ rounds, on corresponding input registers $A_{in}^1, B_{in}^1, A_{in}^2, B_{in}^2$ until $\Pi^2$ has finished. 
\item Finish running protocol $\Pi^1$
\item Take as output the output registers $A_{out}^1, B_{out}^1, A_{out}^2, B_{out}^2$ of both $\Pi^1$ and $\Pi^2$.
\end{enumerate}
%\caption{Protocol $\Pi^1 \otimes \Pi^2$}
%\label{Protocol2}
%\end{Protocol}

It is clear that the channel that $\Pi_2$ implements is $\Pi_2 = \Pi^1 \otimes \Pi^2$,
and the number of rounds satisfies $r = \max (r_1, r_2)$,
so is left to analyze its quantum information cost on input $ \rho_{12}$.
Let $R_{12}$ be a purifying register such that $\rho_{12}^{A_{in}^1 B_{in}^1 A_{in}^2 B_{in}^2 R_{12}}$ is a pure state.
Also, denote the purified joint state in round $i$ as $(\rho_{12}^i)^{A_i^1 B_i^1 C_i^1 A_i^2 B_i^2 C_i^2 R_{12}}$, and the local state for protocol $\Pi^1$ as 
\begin{align}
	(\rho_1^i)^{A_i^1 B_i^1 C_i^1} = \Tr{A_i^2 B_i^2 C_i^2 R_{12}}{(\rho_{12}^i)^{A_i^1 B_i^1 C_i^1 A_i^2 B_i^2 C_i^2 R_{12}}},
\end{align}
and similarly for that of protocol $\Pi^2$. Notice that for all $i$, $(\rho_1^i)^{A_i^1 B_i^1 C_i^1}$ is purified by $(\rho_1^i)^{A_i^1 B_i^1 C_i^1 A_{in}^2 B_{in}^2 R_{12} } \otimes \phi_2^{T_A^2 T_B^2}$, with $ A_{in}^2 B_{in}^2 R_{12}$ the registers of state $\rho_{12}$ before application of the unitaries corresponding to $\Pi^1$, and $\phi_2$ is the pure entangled state used in $\Pi_2$.
If we denote, for $i \geq r_2 + 1, A_i^2 = A_{out}^2 \otimes (A^{\prime})^2, B_{i}^2 = B_{out}^2 \otimes (B^{\prime})^2$,
then by the definition of QIC and application of chain rule,
\begin{align*}
		2 \cdot QIC (\Pi_2, \rho_{12})  
		& = \sum_{i = 1, i~odd}^{r_2 } I(C_{i}^1 C_{i}^2 ; R_{12} | B_i^1 B_i^2)_{\rho_{12}} + \sum_{i = 1, i~even}^{r_2 } I(C_{i}^1 C_{i}^2 ; R_{12} | A_i^1 A_i^2)_{\rho_{12}} \\
		&+ \sum_{i = r_2 + 1, i~odd}^{r_1 } I(C_{i }^1  ; R_{12} | B_{i}^1  B_{i}^2)_{\rho_{12}} + \sum_{i = r_2 + 1, i~even}^{r_1 } I(C_{i }^1 ; R_{12} |  A_{i}^1  A_{i}^2)_{\rho_{12}} \\
		& = \sum_{i = 1, i~odd}^{r_2 } I(C_{i}^2 ; R_{12} | B_i^1 B_i^2 C_{i}^1 )_{\rho_{12}} + \sum_{i = 1, i~even}^{r_2 } I(C_{i}^2 ; R_{12} | A_i^1 A_i^2 C_{i}^1)_{\rho_{12}} \\
		& +   \sum_{i = 1, i~odd}^{r_1 } I(C_{i}^1 ; R_{12} | B_i^1 B_i^2)_{\rho_{12}} + \sum_{i = 1, i~even}^{r_1 } I(C_{i}^1 ; R_{12} | A_i^1 A_i^2)_{\rho_{12}}.
\end{align*} 

Now for protocol $\Pi^1$, as noted above, the registers $A_{in}^2 B_{in}^2 R_{12} T_A^2 T_B^2$ purify $(\rho_1^i)^{A_i^1 B_i^1 C_i^1}$ for all $i$, so

\begin{align*}
	2 \cdot QIC (\Pi^1, \rho_{1})  
		& = \sum_{i=1, i~odd}^{r_1 } I(C_{i}^1  ; A_{in}^2 B_{in}^2 R_{12} T_A^2 T_B^2 | B_i^1 )_{\rho_{1}} + \sum_{i=1, i~even}^{r_1 } I(C_{i}^1  ; A_{in}^2 B_{in}^2 R_{12} T_A^2 T_B^2 | A_i^1 )_{\rho_{1}} \\
		& = \sum_{i = 1, i~odd}^{r_1 } I(C_{i}^1  ; A_{i}^2 B_i^2 C_{i}^2 R_{12}  | B_i^1 )_{\rho_{12}} +  \sum_{i = 1, i~even}^{r_1 } I(C_{i}^1  ; A_{i}^2 B_{i}^2 C_{i}^2 R_{12}  | A_i^1 )_{\rho_{12}} \\
		& = \sum_{i=1, i~odd}^{r_1 } I(C_{i}^1  ; B_i^2  | B_i^1 )_{\rho_{12}} + \sum_{i=1, i~even}^{r_1 } I(C_{i}^1  ; A_i^2  | A_i^1 )_{\rho_{12}} \\
 		& + \sum_{i=1, i~odd}^{r_1 } I(C_{i}^1  ;  R_{12}  | B_i^1 B_i^2 )_{\rho_{12}} + \sum_{i=1, i~even}^{r_1 } I(C_{i}^1 ;  R_{12}  | A_i^1 A_i^2 )_{\rho_{12}} \\
		& + \sum_{i=1, i~odd}^{r_1 } I(C_{i}^1  ; A_{i}^2  C_{i}^2   | B_i^1 B_i^2 R_{12} )_{\rho_{12}} + \sum_{i=1, i~even}^{r_1 } I(C_{i}^1 ;  B_{i}^2 C_{i}^2   | A_i^1 A_i^2 R_{12} )_{\rho_{12}} \\
		& \geq \sum_{i=1, i~odd}^{r_1 } I(C_{i}^1  ;  R_{12}  | B_i^1 B_i^2 )_{\rho_{12}} + \sum_{i=1, i~even}^{r_1 } I(C_{i}^1 ;  R_{12}  | A_i^1 A_i^2 )_{\rho_{12}},
\end{align*}
in which the first equality is by definition, the second is by isometric invariance of the conditional quantum mutual information (CQMI), the third by the chain rule for CQMI, and the inequality is by non-negativity of CQMI.
Similarly for protocol $\Pi^2$, with a slightly different application of the chain rule, we get
\begin{align*}
	2 \cdot QIC (\Pi^2, \rho_{2})  
		& = \sum_{i=1, i~odd}^{r_2 }  I(C_{i}^2  ; A_{in}^1 B_{in}^1 R_{12} T_A^1 T_B^1 | B_i^2 )_{\rho_{2}}  + \sum_{i=1, i~even}^{r_2 } I(C_{i}^2  ; A_{in}^1 B_{in}^1 R_{12} T_A^1 T_B^1 | A_i^2 )_{\rho_{2}} \\
		& = \sum_{i=1, i~odd}^{r_2 } I(C_{i}^2  ; A_{i}^1 B_i^1 C_{i}^1 R_{12}  | B_i^2 )_{\rho_{12}} + \sum_{i=1, i~even}^{r_2 } I(C_{i}^2 ; A_i^1 B_{i}^1 C_{i}^1 R_{12}  | A_i^2 )_{\rho_{12}} \\
		& = \sum_{i=1, i~odd}^{r_2 }  I(C_{i}^2  ; B_i^1 C_{i}^1 | B_i^2 )_{\rho_{12}}   + \sum_{i=1, i~even}^{r_2 } I(C_{i}^2 ; A_i^1 C_{i}^1  | A_i^2 )_{\rho_{12}} \\
		& + \sum_{i=1, i~odd}^{r_2 } I(C_{i}^2  ;  R_{12}  | B_i^1 B_i^2 C_{i}^1 )_{\rho_{12}} + \sum_{i=1, i~even}^{r_2 } I(C_{i}^2 ;  R_{12}  | A_i^1 A_i^2 C_{i}^1 )_{\rho_{12}} \\
		& + \sum_{i=1, i~odd}^{r_2 } I(C_{i}^2  ; A_{i}^1     | B_i^1 B_i^2 C_{i}^1 R_{12} )_{\rho_{12}} + \sum_{i=1, i~even}^{r_2 } I(C_{i}^2 ;  B_{i}^2    | A_i^1 A_i^2 C_{i}^1 R_{12} )_{\rho_{12}} \\
		& \geq \sum_{i=1, i~odd}^{r_2 } I(C_{i}^2  ;  R_{12}  | B_i^1 B_i^2 C_{i}^1 )_{\rho_{12}} + \sum_{i=1, i~even}^{r_2 } I(C_{i}^2 ;  R_{12}  | A_i^1 A_i^2 C_{i}^1 )_{\rho_{12}}.
\end{align*}
The result then follows by comparing terms.
\end{proof}

\subsection{Reducing the Error for Functions}

Similarly to communication, it is possible to reduce the error when computing functions without increasing too much the information.

\begin{lemma} \label{lem:error:reduction}
	For any function $f$ and error parameter $\epsilon>0$, the following holds:
\begin{align*}
	QIC (f, \epsilon) \le O \big(\log 1/\epsilon \cdot  QIC(f, 1/3)\big).
\end{align*}
\end{lemma}

\begin{proof}
Given $\delta > 0$, let $\Pi$ be a protocol computing $f$ correctly except with 
probability $1/3$ on every input and satisfying $QIC (\Pi, \mu ) \leq QIC(f, 1/3) + \delta$ for 
all $\mu$. Let $n \in O (\log 1/\epsilon)$ be given by the Chernoff bound such that 
protocol $\Pi_n$ running $\Pi$ $n$ times in parallel as per Lemma~\ref{lem:subadd}, 
with each input being a copy of the instance to $f$, and taking a majority 
vote (with arbitrary tie-breaking) computes $f$ correctly except with 
probability $\epsilon$ on every input. 
This $n$ can be chosen independently of $\delta$.
We now argue on the quantum 
information cost of $\Pi_n$.
Consider an arbitrary distribution $\mu$ for $f$, and let $\mu_n$ be 
the distribution once the $n$ copies have been made.
If we denote the marginal for the $i$-th copy by $\mu^i$, 
then $\mu^i = \mu$. By Lemma \ref{lem:subadd} and an easy induction, we then get that 
\begin{align*}
	QIC (f, \epsilon) & \leq QIC (\Pi_n, \mu_n) \\
		& \leq n QIC (\Pi, \mu) \\
		&\leq n (QIC (f, 1/3) + \delta).
\end{align*}
 The result follows by taking $\delta$ to $0$.
\end{proof}

\subsection{Reduction from DISJ to AND}

With the following definition, the above proof also establishes the following corollary.

\begin{definition}
For all $r \in \mathbb{N}, \epsilon \in [0, 1]$,
\begin{align*}
	QIC_0^r (AND, \epsilon) = \inf_{\Pi \in \T^r (AND, \epsilon)} \max_{\mu_0} QIC (\Pi, \mu_0),
\end{align*}
in which the maximum ranges over all $\mu_0$ satisfying $\mu_0 (1, 1 ) = 0$.
\end{definition}

\begin{corollary}
\label{cor:reduceAND}
	For any $\epsilon>0$ and $r \in \mathbb{N}$,
\begin{align*}
	QIC_0^r (AND, \epsilon) \le O \big(\log 1/\epsilon \cdot  QIC_0^r (AND, 1/3)\big).
\end{align*}
\end{corollary}

We provide a slight variant of the argument of \cite{Tou14} to obtain a low information protocol for AND from a protocol for disjointness.

\begin{lemma} \label{disj_to_and}
For any $n, r, \epsilon$ and $\mu_0$ such that $\mu_0 (1, 1) = 0$,
\begin{align*}
	 \inf_{\Pi_{A} \in \T^r (AND, \epsilon)}  QIC (\Pi_{A}, \mu_0) \leq  \inf_{\Pi_{D} \in \T^r (DISJ_n, \epsilon)} \frac{1}{n} QIC (\Pi_D, \mu_0^{\otimes n}).
\end{align*}
\end{lemma}

\begin{proof}
Let $I_n = \inf_{\Pi_{D} \in \T^r (DISJ_n, \epsilon)} QIC (\Pi_D, \mu_0^{\otimes n})$. We prove the result by induction on $n$. The base case is trivial since $DISJ_1 = \neg AND$, and so a protocol to compute $DISJ_1$ with error $\epsilon$ can be used to compute $AND$ with error $\epsilon$ and vice-versa. In particular, we get $I_1 = \inf_{\Pi_{A} \in \T^r (AND, \epsilon)}  QIC (\Pi_{A}, \mu_0)$.
For the induction, suppose the result holds for $DISJ_{n-1}$, we will use Lemma~\ref{lem:add2} to go from $DISJ_n$ to $DISJ_1$ and $DISJ_{n-1}$. Indeed, given $\delta > 0$ and $\Pi_D$ computing $DISJ_n$ with error $\epsilon$ and satisfying $QIC (\Pi_D, \mu_0^{\otimes n}) \leq I_n + \delta$, we can use Lemma~\ref{lem:add2} with $\rho_1 = \mu_0, \rho_2 = \mu_0^{\otimes n-1}$ and then  it is clear that $\Pi^1$ computes $DISJ_1$ with error $\epsilon$ and $\Pi^2$ computes $DISJ_{n-1}$ with error $\epsilon$. We get
\begin{align*}
	I_n + \delta & \geq QIC (\Pi_D, \mu_0^{\otimes n}) \\
			& = QIC (\Pi^1, \mu_0) + QIC (\Pi^2, \mu_0^{\otimes n-1}) \\
			& \geq I_1 + I_{n-1} \\
			& \geq n I_1.
\end{align*}
\end{proof}

The following lemma is very similar to Theorem \ref{thm:minimax}. The only difference is that the distributions we consider are restricted and on the right hand side the error of the protocol is measured in the worst case. Since the error is worst case, there is no loss in the error, and the payoff function would be simply $P_B(\Pi, \mu) = QIC(\Pi, \mu)/I$. 

\begin{lemma}
\label{andminimax}
\[
QIC_0^r (AND, \epsilon) =  \max_{\mu_0,\mu_0 (1, 1 ) = 0} \inf_{\Pi \in \T^r (AND, \epsilon)} QIC (\Pi, \mu_0)
\]
\end{lemma}

\begin{lemma}
\label{lem:disjvsand1}
For all $r, n \in \mathbb{N}$,
\begin{align*}
	QCC^r (DISJ_n, 1/3) \geq n \cdot QIC_0^r (AND, 1/3)
\end{align*}
\end{lemma}

\begin{proof}
The result follows from the following chain of inequality:
\begin{align*}
	QCC^r (DISJ_n, 1/3) &\geq QIC^r (DISJ_n, 1/3) \\
				& \geq \max_{\mu_0} \inf_{\Pi_D \in \T^r (DISJ_n, 1/3)} QIC (\Pi_D, \mu_0^{\otimes n}) \\
				& \geq \max_{\mu_0} \inf_{\Pi_A \in \T^r (AND, 1/3)} n \cdot QIC (\Pi_A, \mu_0) \\
				& \geq n \cdot QIC_0^r (AND, 1/3).
\end{align*}
The first inequality is by Lemma~\ref{lem:qicvsqccpf}, the second since,
on the r.h.s., the maximization is over a smaller set of product distributions with $\mu_0 (1, 1) = 0$ 
and the minimization over a larger set of protocols, the third is by Lemma~\ref{disj_to_and}, and the last is by Lemma~\ref{andminimax}.
\end{proof}

\ignore{
\section{SQDPT}
\jiemingnote{need to be moved to prelim}
\begin{theorem}{\label{thm:sdp_sherstov}}(\cite{sdp_sherstov}) Let $\eps_{\text{sh}} > 0$ be a small enough absolute constant. Then for every sign matrix $F$ , the following communication problem requires $\Omega(n GDM_{1/5}(F))$ qubits of communication $($with arbitrary entanglement$)$: Solving with probability $2^{- \eps_{\text{sh}} n}$, at least $(1-\eps_{\text{sh}}) n$ among $n$ instances of $F$. 
\end{theorem}
}

\section{Lower bound on QIC by generalized discrepancy method}

\subsection{Compression}

\begin{definition} We say that $QCC(f^k, \mu^k, \eta_1 k, \eta_2) \le C$ if there exists a protocol $\pi$ for $f^k$ s.t. $QCC(\pi) \le C$ and 
$$
Pr[\text{$\pi$ computes $\ge \eta_1 k$ coordinates correctly}] \ge 1 - \eta_2
$$
Here the probability is both over the distribution $\mu^k$ and the randomness of protocol (which includes the randomness due to quantum measurements). We don't require the protocol to declare which coordinates were computed correctly. 
\end{definition}

\begin{lemma}{\label{compression}} If there exists a protocol $\Pi$ for $f$ with error $\le \eps$ w.r.t $\mu$ s.t. $QIC(\Pi, \mu) = I$, then for all $\eps', \delta > 0$, there exists $k_0(\Pi,\mu, \eps',\delta)$ such that for all $k \ge k_0$, $QCC(f^k, \mu^k, (1- 2\eps)k, e^{-2\eps^2 k} + \eps') \le k (I + \delta)$.
\end{lemma}

\begin{proof} Suppose $(E_1,\ldots, E_k)$ is the vector of indicator random variables of the errors in various coordinates of $\Pi^{\otimes k}$ i.e. $E_i = 1$ if error occurred on the $i^{\text{th}}$ coordinate. Also look at $\Pi_k$ obtained from lemma \ref{lem:coding} for large enough $k$ with parameters $2\eps'$, $\delta$ and where $\rho$ is  $\mu$. Suppose $(E^{'}_1,\ldots, E^{'}_k)$ is the vector of errors for $\Pi_k$. According to lemma \ref{lem:coding}, $\Pi_k$ satisfies the following: 
\begin{align*}
\mathbb{E}_{((x_1,\ldots,x_k), (y_1,\ldots,y_k)) \sim \mu^{k}} ||\Pi_k((x_1,\ldots,x_k), (y_1,\ldots,y_k)) - \Pi^{\otimes k}((x_1,\ldots,x_k), (y_1,\ldots,y_k))||_1 \le 2 \eps'
\end{align*}
Hence it follows that 
$$
||(E_1,\ldots, E_k) -  (E^{'}_1,\ldots, E^{'}_k)||_{\text{TV}} \le \eps'
$$
Here $||P - Q||_{\text{TV}}$ is the total variation distance between the distributions $P$ and $Q$ (we are not distinguishing between random variables and their distributions). Since $\Pr[\sum_i E_i \ge 2 \eps k] \le e^{-2\eps^2 k}$ by Chernoff bounds, it follows that 
$$
\Pr \left[\sum_i E^{'}_i \ge 2 \eps k \right] \le e^{-2\eps^2 k} + \eps'
$$
which implies the lemma along with the fact that $QCC(\Pi_k) \le (I+\delta)k$.
\end{proof}

\subsection{Average case to worst case}

In this section, we prove the following lemma which turns a protocol for average case input to a protocol for worst case input.
\begin{lemma}
\label{atw}
Suppose $f_n : \{0,1\}^n \times \{0,1\}^n \rightarrow \{0,1\}$ is an arbitrary boolean function. Let $k \geq 2^{5n}$ and $\epsilon > 10k^{-0.005}$. Assume for any product input distribution $\mu^k$, there exists a protocol $\pi_{\mu^k}$ with $QCC\left(\pi_{\mu^k}\right) \leq l$ that computes at least $\left(1-\alpha\right)k$ coordinates of $f_n^k$ correctly with probability at least $\gamma$. Then there exists a protocol $\tau$ s.t. for any input $((x_1,\cdots,x_k),(y_1,\cdots,y_k))$, for any integer $c \ge 3$ and constant $\epsilon > 0$,  $\tau$ computes at least $\left(1-2^{-c/2}-c\alpha\right)k$ coordinates of $f_n^k$ correctly with probability at least $\frac{1}{2}\left(\left(\frac{\gamma}{\left(1+\epsilon\right)^k}\right)^c -2^{-2^{2-2c}k}\right)$. Also $QCC\left(\tau\right) \leq c \cdot l  + o\left(k\right)$. 
\end{lemma}

\begin{proof}
In this lemma, we want to construct a protocol $\tau$ which works for an arbitrary input based on protocols which work on product input distributions (product across coordinates). The main idea of the proof is that corresponding to any input $(\left(x_1,...,x_k\right)$, $\left(y_1,...,y_k\right))$ ($x_i$ and $y_i$ are inputs of a $f_n$ instance and have $n$ bits), we can associate a $\mu$, which is the empirical distribution:
\[
\mu\left(x,y\right) = \frac{\# \text{ of } i, \left(x_i,y_i\right) = \left(x,y\right)}{k}.
\]

\ignore{
It is clear that Alice and Bob can sample their own parts of the input from distribution $\mu^k$ by using public randomness. In protocol $\tau$ we will run $\pi_{\mu^k}$ on distribution $\mu^k$ for $c$ times. For $\left(x_i,y_i\right)$ that is covered in this $c$ times running of $\pi_{\mu^k}$, Alice and Bob can get the answer from running $\pi_{\mu^k}$. For the rest part of the input, we will prove that with high probability, the number of input left is less than $2^{c/2}$.  One problem of this protocol is that although Alice and Bob can sample from $\mu^k$ without communication by using public randomness, both of them don't have the full information distribution $\mu^k$. And thus they cannot run protocol $\pi_{\mu^k}$.
}

\noindent So it makes sense to construct $\tau$ from $\pi_{\mu^k}$.  The players can simulate $\mu^k$ by sampling independent coordinates from their input (with replacement). However the issue is that the players don't know $\mu$, so they have no idea what $\pi_{\mu^k}$ is. So in the actual protocol Alice and Bob will first sample some coordinates to get an estimate $\tilde{\mu}$ of $\mu$ and then run protocol $\pi_{\tilde{\mu}^k}$. The protocol $\tau$ is described in Protocol \ref{Protocol:tau}. 

\begin{Protocol}
Inputs: $(x_1,\ldots, x_k)$ and $(y_1,\ldots, y_k)$
\begin{enumerate}
\item Get an estimate $\tilde{\mu}$ of $\mu$. 
\item Alice and Bob use shared randomness to obtain random independent samples from $[k]$, $j_1,\ldots, j_{ck}$. Run the protocol $\pi_{\tilde{\mu}^k}$ $c$ times. In the $t^{\text{th}}$ iteration, the protocol is run on inputs $(x_{j_{(t-1)k + 1}}, \ldots, x_{j_{tk}}),(y_{j_{(t-1)k + 1}}, \ldots, y_{j_{tk}})$. In the process we obtain answers for various coordinates (some of the coordinates will be sampled multiple times and we will obtain multiple answers for them).
\item If a coordinate was sampled in the previous step, output the answer $\pi_{\tilde{\mu}^k}$ gave for it. If they got multiple results on one coordinate, they will output the first one. If a coordinate was not sampled, output $0$ on that coordinate.
\end{enumerate}
\caption{Protocol $\tau$}
\label{Protocol:tau}
\end{Protocol}

%\item If the number of $i\in[k]$ such that $\left(x_i,y_i\right)$ is not sampled in the previous step is more than $2^{-c/2}k$, then fail. 

Now let's analyze this protocol. We first need the following two lemmas to show how to get an estimate $\tilde{\mu}$ of $\mu$. 
\begin{lemma}
\label{qcnt}
After communicating $O(k^{0.52}\log k)$ bits, for some specific input $(x,y)$, with success probability at least $1-1/k$, Alice and Bob know $\mu(x,y)$ exactly if $\mu(x,y) \cdot k  < k^{0.02}$, otherwise Alice and Bob know that $\mu(x,y) \cdot k \geq k^{0.02}$. 
\end{lemma}
\begin{proof}
In \cite{BCW}, they showed that to compute the disjointness between two inputs of length $k$, the quantum communication complexity is $O(\sqrt{k} \log k)$. The corresponding protocol has constant error rate and will find one intersection place. We will use this protocol to solve our problem by the following reduction. For each input $(x_i,y_i)$, we set $a_i = 1_{x_i = x}$ and $b_i = 1_{y_i = y}$. Then finding $(x,y)$ in the input is just like finding intersection between $a = (a_1, ... ,a_k)$ and $b=(b_1, ... ,b_k)$. Protocol \ref{Protocol:count} shows how to finish the task described in the lemma.

\begin{Protocol}
\begin{enumerate}
\item Set $a$ and $b$ as we just described. Set $cnt = 0$. 
\item Do the following step $c_1 \cdot k^{0.02}$ times, $c_1$ is some constant to be figured out in the proof:
\item Use protocol for DISJ in \cite{BCW} to find the intersection between $a$ and $b$, let it be at place $j$, Alice and Bob communicate 2 bits to check if $a_j = b_j = 1$. If it is true, then set $cnt = cnt + 1$, $a_j = 0$, $b_j = 0$. 
\end{enumerate}
\caption{Protocol count}
\label{Protocol:count}
\end{Protocol}

Let's analyze this protocol. First its quantum communication cost is clear to be $O(k^{0.52}\log k)$ as the DISJ protocol has quantum communication cost $O(\sqrt{k} \log k)$. Then for each repeat of step 3, if the DISJ protocol gives wrong answer, we will not do anything. And if the DISJ protocol gives the correct intersection, the counter will be increased by one and the intersection place will be removed and we can find other intersections. Thus we only have to show with probability at least $1- 1/k$, DISJ protocol gives a correct answer for at least $k^{0.02}$ times. Assume the DISJ protocol succeeds with some constant probability $p$. Let $Cr$ denote the random variable for the number of correct answers DISJ protocol gives. We know $\mathbb{E}[Cr] = p\cdot c_1 \cdot k^{0.02} $. By the additive Chernoff bound, the probability that DISJ protocol give a correct answer for at least $k^{0.02}$ times is 
\[
\Pr[Cr \geq k^{0.02}]  = 1 - \Pr[Cr < k^{0.02}]  \geq 1 - e^{- 2 (p\cdot c_1 \cdot k^{0.02} - k^{0.02})^2 / (c_1 \cdot k^{0.02})}.
\]
By picking $c_1$ properly, for example $c_1 = 2/p$, we get $\Pr[Cr \geq k^{0.02}] \geq 1 - 1/k$. 
\end{proof}

\begin{lemma}
\label{sml}
Let $\epsilon > 10k^{-0.005}$ be some constant. After communicating $O\left(k^{0.99}\cdot n + 2^{2n}\cdot k^{0.52} \log k \right)$ bits, with probability at least $1/2$, Alice and Bob agree on some $\tilde{\mu}$, such that for any $\left(x,y\right)$, $\frac{\tilde{\mu}\left(x,y\right)}{\mu\left(x,y\right)}  < 1 + \epsilon$. 
\end{lemma}
\begin{proof}
We use the following protocol to estimate $\mu$:
\begin{Protocol}
Inputs: $(x_1,\ldots, x_k)$ and $(y_1,\ldots, y_k)$
\begin{enumerate}
\item Sample the coordinates randomly $k^{0.99}$ times using public randomness (with replacement). Alice and Bob exchange their input for these coordinates. For each $(x,y) \in \{0,1\}^n \times \{0,1\}^n$, count the number of times it appears in these coordinates and denote the count by $c_1(x,y)$.
\item For all $\left(x,y\right)$, use Lemma \ref{qcnt} to count the number of times $\left(x,y\right)$ appears in the input and denote the count obtained by $c_2\left(x,y\right)$. 
\item We combine $c_1$ and $c_2$ as $c_3$. For each $\left(x,y\right)$, if $c_2\left(x,y\right)\geq k^{0.02}$, let $c_3\left(x,y\right) = c_1\left(x,y\right) \cdot k^{0.01}$ otherwise $c_3\left(x,y\right) = c_2\left(x,y\right)$. 
\item $\tilde{\mu}\left(x,y\right) = \frac{c_3\left(x,y\right)}{\sum_{x',y'}c_3\left(x',y'\right)}$.
\end{enumerate}
\caption{Estimate $\mu$}
\end{Protocol}

Let's first analyze the communication cost of this part. It's clear that the first step needs at most $O\left(k^{0.99}n\right)$ communication. For second step, by Lemma \ref{qcnt}, it needs at most $O\left(2^{2n}\cdot k^{0.52}\log k \right)$ communication. Sum them up, this protocol needs $O\left(k^{0.99}\cdot n + 2^{2n}\cdot k^{0.52}\log k\right)$ bits of communication.

Then let's consider the following events:
\begin{enumerate}
\item For all $\left(x,y\right)$ such that $\mu\left(x,y\right)\cdot k  \geq k^{0.02} $, $|c_1\left(x,y\right) \cdot k^{0.01} - \mu\left(x,y\right) \cdot k| < \frac{\epsilon}{3} \mu\left(x,y\right) \cdot k$.
\item For any $(x,y)$, the protocol described in Lemma \ref{qcnt} does not fail.
\end{enumerate}

If these two events happen, then we know that $|c_3\left(x,y\right) - \mu\left(x,y\right) \cdot k| < \frac{\epsilon}{3} \mu\left(x,y\right) \cdot k$, therefore as desired,
\[
\tilde{\mu}\left(x,y\right) = \frac{c_3\left(x,y\right)}{\sum_{x',y'}c_3\left(x',y'\right)}  \leq  \frac{\left(1+\frac{\epsilon}{3}\right) \mu\left(x,y\right) \cdot k}{\left(1-\frac{\epsilon}{3}\right)\cdot k} < \left(1 + \epsilon\right)\mu\left(x,y\right).
\]
Finally, we only have to make sure that these two events happen with probability at least $1/2$. For the first event, by the multiplicative Chernoff bound and union bound, it does not happen with probability 
\[
2^{2n} \cdot \Pr[|c_3\left(x,y\right) / k^{0.99} - \mu\left(x,y\right)| > \frac{\epsilon}{3} \mu\left(x,y\right)] < 2ke^{-\frac{\left(\epsilon/3\right)^2\mu\left(x,y\right)k^{0.99}}{3}} \leq 2ke^{-\epsilon^2 k^{0.01}/27} < 1/4.
\]
For the second event, by Lemma \ref{qcnt} and the union bound, it does not happen with probability at most $2^{2n} \cdot \frac{1}{k} <  1/4$. Thus these two events happen with probability at least $1/2$.
\end{proof}

\noindent Let's consider the communication cost of $\tau$. For the first step, the cost is  $O\left(k^{0.99}\cdot n + 2^{2n}\cdot k^{0.52}\log k\right)= o\left(k\right)$. For the second step, the quantum communication complexity is at most $c\cdot l$. For the third step, the cost is 0. Therefore  $QCC\left(\tau\right) \leq c \cdot l  + o\left(k\right)$. 

\noindent Let's say that the protocol $\tau$ succeeds when the following things happen:
\begin{enumerate}
\item For all $\left(x,y\right)$, $\frac{\tilde{\mu}\left(x,y\right)}{\mu\left(x,y\right)}  < 1 + \epsilon$. 
\item The $c$ runs of protocol $\pi_{\tilde{\mu}^k}$ in step 2 of protocol $\tau$ all compute at least $\left(1-\alpha\right)k$ coordinates correctly.
\item Number of $i\in[k]$ such that the coordinate $i$ is not sampled in step 2 of protocol $\tau$ is at most $2^{-c/2}k$.
\end{enumerate}
If $\tau$ succeeds, then it computes at least $\left(1-2^{-c/2}-c\alpha\right)k$ coordinates correctly. This is because errors come from two possible ways: 
\begin{enumerate}
\item Some coordinates are not sampled. When $\tau$ succeeds, the number of coordinates that are not sampled is at most $2^{-c/2}k$.
\item Some coordinates' results are wrong in step 2. When $\tau$ succeeds, the number of errors from step 2 is at most $\alpha ck$.
\end{enumerate}

Finally, let's analyze the success probability of protocol $\tau$. Let's analyze step by step:
\begin{enumerate}
\item For step one, by Lemma \ref{sml}, it is clear that we succeed with probability $1/2$. 
\item For step two, first we know that when running $\pi_{\tilde{\mu}^k}$ on distribution $\tilde{\mu}^k$, we succeed with probability at least $\gamma$. And since we have for any $\left(x,y\right)$, $\frac{\tilde{\mu}\left(x,y\right)}{\mu\left(x,y\right)}  < 1 + \epsilon$, if we run  $\pi_{\tilde{\mu}^k}$ on distribution $\mu^k$, the success probability will be at least $\frac{\gamma}{\left(1+\epsilon\right)^k}$. When running this protocol $c$ times independently, the success probability will be at least $\left(\frac{\gamma}{\left(1+\epsilon\right)^k}\right)^c$. Note that when we sample coordinates independently at random, the distribution we induce is $\mu^k$. 
\item It is only left to analyze the probability that number of coordinates not sampled in step 2 of protocol $\tau$ is at least $2^{-c/2}k$. For each coordinate $i$, define $s_i$ to be the random variable that indicates whether coordinate $i$ is sampled or not ($1$ means not sampled and $0$ means sampled). Then we have $\mathbb{E}[s_i]$$= \left(1-\frac{1}{k}\right)^{ck} < 2^{-c}$. In order to show the failure probability small by Chernoff bound, we will show that all the $s_i$'s are negatively correlated. To show they are negatively correlated, we only have to show
\[
\forall I \subseteq [k], \Pr\left[\prod_{i\in I} s_i = 1\right] \leq \prod_{i\in I}\Pr[s_i = 1].
\]
Notice that $\Pr\left[\prod_{i \in I} s_i = 1\right] = \left(1-\frac{|I|}{k}\right)^{kc}$ and $\Pr[s_i = 1] = \left(1- \frac{1}{k}\right)^{kc}$. So we have,
\[
\forall I \subseteq [k], \Pr\left[\prod_{i\in I} s_i = 1\right]  = \left(1-\frac{|I|}{k}\right)^{kc} \leq  \left(\left(1-\frac{1}{k}\right)^{|I|}\right)^{kc} = \prod_{i\in I}\Pr[s_i = 1].
\]
Since all the $s_i$'s are negatively correlated, by Chernoff bound for negatively correlated random variables, for example see \cite{concentration_reference}, we have that the failure probability 
\[
\Pr\left[\sum_{i=1}^k s_i  \geq 2^{-c/2}k\right] < e^{-2k \left(2^{-c/2} - 2^{-c}\right)^2} < e^{-2^{2-2c}k} < 2^{-2^{2-2c}k}.
\]
\end{enumerate}
The second inequality holds for all $c \ge 3$. Notice that the event that we err in the first step is independent from the event that we err in the second step. So the success probability of $\tau$ is at least $\frac{1}{2}\left(\left(\frac{\gamma}{\left(1+\epsilon\right)^k}\right)^c -2^{-2^{2-2c}k}\right)$.

\end{proof}

\subsection{Lower bound on QIC}

\begin{definition} We say that $QCC(f^k, \eta_1 k, \eta_2) \le C$ if there exists a protocol $\pi$ for $f^k$ s.t. $QCC(\pi) \le C$ and 
$$
Pr[\text{$\pi$ computes $\ge \eta_1 k$ coordinates correctly}] \ge 1 - \eta_2
$$
Here the probability is over randomness of protocol (which includes the randomness due to quantum measurements). We don't require the protocol to declare which coordinates were computed correctly. 
\end{definition}

\begin{theorem} 
\label{lem:QICvsGDM}
There exists an absolute constant $\eta > 0$ s.t. for any boolean function $f$, $QIC_D(f, \eta)$ $\ge \Omega(GDM_{1/5}(f) - O(1))$.
\end{theorem}

\begin{proof}
Let $\eta > 0$ be a sufficiently small constant to be fixed later. Suppose $\max_{\mu} QIC(f, \mu, \eta) = I$. We will show that for sufficiently large $k$, it holds that
$$
QCC(f^k, (1-\eps_{\text{sh}}) k, 1 - 2^{-\eps_{\text{sh}} k}) \le O( k \cdot (I + 2)) + o(k)
$$
from which the theorem follows from Theorem \ref{thm:sdp_sherstov}. 
\\
\\
By definition, for all $\mu$, there exists a protocol $\Pi_{\mu}$ for $f$ s.t. $QIC(\Pi_{\mu}, \mu) \le I + 1$ and error $\le \eta$ w.r.t $\mu$. By lemma \ref{compression}, for sufficiently large $k$, there exists a protocol $\Pi_{k, \mu, \eps'}$ s.t. $QCC(\Pi_{k, \mu, \eps'}) \le k (I + 2)$ and 
$$
\Pr[\text{$\Pi_{k, \mu, \eps'}$ computes $\ge (1 - 2 \eta) k$ coordinates of $f^k$ correctly}] \ge 1 - e^{-2\eta^2 k} - \eps'
$$
Here the probability is over the distribution $\mu^k$ and the randomness of the protocol. Choose $k$ large enough and $\eps'$ small enough so that $1 - e^{-2\eta^2 k} - \eps' \ge 0.9$. Then by lemma \ref{atw}, for any integer $c > 0$, any constant $\eps > 0$, there exists a protocol $\tau$ s.t. 
\begin{align*}
&\Pr[\text{$\tau$ computes $\ge (1 - 2^{-c/2} - 2c\eta) k$ coordinates correctly (on any input $(x_1,\ldots,x_k, y_1,\ldots, y_k)$)}] \\ 
&\ge \frac{1}{2}\left( \left(\frac{0.9}{(1+\epsilon)^k}\right)^c -2^{-2^{2-2c}k} \right)
\end{align*}
Here the randomness is only over the randomness of the protocol. Also $QCC(\tau) \le c \cdot k \cdot (I+2) + o(k)$. Choose $c = \lceil 2 \log \left( \frac{2}{\eps_{\text{sh}}}\right) \rceil$. Also choose $\eta = \frac{\eps_{\text{sh}}}{4 c}$. Then 
$$
1 - 2^{-c/2} - 2c\eta \ge 1 - \eps_{\text{sh}}
$$
Since $2^{2x} \ge 1 + x$ for all $x > 0$, it follows that
$$
\left(\frac{0.9}{(1+\epsilon)^k}\right)^c \ge 0.9^c \cdot 2^{- 2 \cdot \eps \cdot c \cdot k} \ge 2^{ - (2 \eps k + 1) \cdot c} \ge 2^{- 4 \cdot \eps \cdot c \cdot k}
$$
The last inequality is true for sufficiently large $k$. Now choose $\eps = \eps_{\text{sh}}^4/100c$. Then since 
$$
2^{-2^{2-2c}k} \le 2^{- \eps_{\text{sh}}^4 k/16}
$$
we get that
\begin{align*}
\frac{1}{2}\left( \left(\frac{0.9}{(1+\epsilon)^k}\right)^c -2^{-2^{2-2c}k} \right) &\ge \frac{1}{2} \left( 2^{- \eps_{\text{sh}}^4 k/25} - 2^{- \eps_{\text{sh}}^4 k/16} \right) \\
&\ge 2^{- \eps_{\text{sh}}^4 k/16} \\
&\ge 2^{ - \eps_{\text{sh}} k}
\end{align*}
The second inequality holds for sufficiently large $k$. Hence $QCC(\tau) \le c \cdot k \cdot (I+2) + o(k)$ and
\begin{align*}
&\Pr[\text{$\tau$ computes $\ge (1 - \eps_{\text{sh}}) k$ coordinates correctly (on any input $(x_1,\ldots,x_k, y_1,\ldots, y_k)$)}] \\ 
&\ge 2^{- \eps_{\text{sh}} k}
\end{align*}
which implies that $QCC(f^k, (1-\eps_{\text{sh}}) k, 1 - 2^{-\eps_{\text{sh}} k}) \le O( k \cdot (I + 2)) + o(k)$. 
\end{proof}

\begin{corollary} \label{qcc:qic}
For all boolean functions $f$, $QCC(f, 1/3) \le 2^{O(QIC(f,1/3) + 1)}$.
\end{corollary}

\begin{proof}
We will use the following folklore result:

\begin{align*}
R(f,1/3) \le \left( \frac{1}{\text{disc}(f)} \right)^{O(1)}
\end{align*}
where $R(f,1/3)$ is the (public-coin) randomized communication complexity of $f$ with error $1/3$ and $\text{disc}(f) = \min_{\mu} \text{disc}^{\mu}(f)$. See, for example, exercise 3.32 in \cite{KushilevitzN97}. This implies 

\begin{align}
QCC(f,1/3) \le R(f,1/3) \le \left( \frac{1}{\text{disc}(f)} \right)^{O(1)} \le 2^{O(GDM_{1/5}(f))} \label{eqn:100}
\end{align}

Now, by theorem \ref{lem:QICvsGDM} and theorem \ref{minimax_general}, we get that $QIC(f,\eta) \ge \Omega(GDM_{1/5}(f) - O(1))$ for some small constant $\eta$. By lemma \ref{lem:error:reduction}, we also get that $QIC(f,1/3) \ge \Omega(GDM_{1/5}(f) - O(1))$, which combined with equation (\ref{eqn:100}) completes the proof.
\end{proof}

\section{From AND to Disj}

In this section, we show that a protocol with low quantum information cost for $AND$ implies a protocol with low quantum information cost for Disjointness

%%%%%%%%%%%%%%THE MAIN LEMMA: FROM AND TO DISJ

\begin{lemma}
\label{lem:DISJvsAND2}
\begin{equation}
\max_{\nu} QIC(DISJ_n, \nu, 2/n) \le n\cdot QIC^r_0(AND, 1/n^2) + O( r \cdot \log^5(n)) + o(\sqrt{n}) \label{eq:mainlemma}
\end{equation}
\end{lemma}

\begin{proof}
Let $QIC^r_0(AND, 1/n^2) = I$. Suppose $\pi$ is a protocol for AND which has error $\le 1/n^2$ for all inputs and s.t. $\max_{\text{$\mu$ s.t. $\mu(1,1) = 0$}} QIC(\pi, \mu) \le I + \delta$, for arbitrary small $\delta$. Using $\pi$, we will construct a protocol for $DISJ_n$. The protocol will have low information cost w.r.t. any distribution $\nu$. Suppose $\tau_k$ is a quantum protocol for $DISJ_k$ that has worst case error $\le 1/k^{10}$ and communication cost $O(\sqrt{k} \log(k))$. For example, use the protocol from \cite{AA03} and amplify the error to $1/k^{10}$. We'll drop the subscript $k$ when it is clear from the context. Consider the protocol $\pi_n$ described as Protocol \ref{Protocol:subsample}.

\begin{Protocol}
Inputs: $(x,y) \in \{0,1\}^n \times \{0,1\}^n$, $(x,y) \sim \nu$ \\
Goal: check if $DISJ_n(x,y) = 1$ or not.
\begin{enumerate}
\item Alice and Bob share a maximally entangled state $\phi_S^{S_A S_B}$ that will serve as 
shared randomness in order to sample uniformly at random $n/\log^{3}(n)$ coordinates from $[n]$ (with replacement). Alice has the register $S_A$ and Bob has $S_B$. 
\item On the random coordinates, run $\tau$. Suppose $O_A$ is the output register for Alice and $O_B$ is the output register for Bob. Note that all this can be implemented using unitaries. Also note either $O_A = O_B = 1$ or $O_A = O_B = 0$. 
\item If $O_A = O_B = 1$, then run $\pi$ on each coordinate. If $\pi$ outputs $1$ on any coordinate, then output $0$, otherwise output $1$. If $O_A = O_B = 0$, Alice and Bob will keep running a dummy protocol (for example keep exchanging a freshly prepared register $\ket{0}$ of dimension same as to be sent in $\pi^n$ in the corresponding step). In the end they output $0$. 
\end{enumerate}
\caption{Subsampling Protocol $\pi_n$}
\label{Protocol:subsample}
\end{Protocol}

We'll denote the protocol in which $\pi$ is run independently on each coordinate by $\pi^n$. First lets analyze the error of the protocol $\pi_n$. Suppose $(x,y)$ were disjoint. Then probability that we output $0$ because of $\tau$ is at most $\log^{30}(n)/n^{10} \le 1/n$. And the probability that we output $0$ because of $\pi^n$ is at most $n/n^2 = 1/n$ because of union bound. So error in this case $\le 2/n$. If the sets were intersecting, even if we don't output $0$ because of $\tau$, we will output $0$ because of $\pi^n$ w.p. at least $1-1/n^2$ (because on the intersecting coordinate, $1/n^2$ is the probability of failure). So in both cases, probability of error $\le 2/n$. 

Now lets figure out the information cost of $\pi_n$. For running $\tau$, we just bound the information cost by communication cost, which is at most $\sqrt{n}/\sqrt{\log(n)} = o(\sqrt{n})$. The interesting part is what happens after $\tau$. Lets look at the state of Alice and Bob after $\tau$ is over. Alice holds the registers $A_{\tau}, O_A, S_A$, where $A_{\tau}$ is what is left behind with Alice after $\tau$, $O_A$ is Alice's output register for $\tau$ and $S_A$ is the entanglement register which acts as shared randomness. Similarly Bob holds $B_{\tau}, O_B, S_B$. After running $i$ steps of $\pi^n$ (just before the $(i+1)^{\text{th}}$ message is transmitted), Alice and Bob hold registers $A_{i+1}$ and $B_{i+1}$ respectively, with $C_{i+1}$ (the register to be sent next) with Alice if $i$ even and with Bob if $i$ odd. Note that the number of rounds of $\pi$ is $r$. Then the information cost of step 3 is:

\begin{align*}
&\frac{1}{2} \cdot \sum_{i = 0, i~even}^{r-1} I(C_{i+1} ; R | B_{i+1},  B_{\tau}, O_B, S_B) + \frac{1}{2} \cdot \sum_{i = 0, i~odd}^{r-1} I(C_{i+1} ; R | A_{i+1}, A_{\tau}, O_A, S_A) \\
&\le \frac{1}{2} \cdot \sum_{i = 0, i~even}^{r-1} I(C_{i+1} ; R, B_{\tau}, O_B, S_B | B_{i+1}) + \frac{1}{2} \cdot \sum_{i = 0, i~odd}^{r-1} I(C_{i+1} ; R, A_{\tau}, O_A, S_A | A_{i+1}) \\
&\le \frac{1}{2} \cdot \sum_{i = 0, i~even}^{r-1} I(C_{i+1} ; R, B_{\tau}, O_B, S_B, A_{\tau}, O_A, S_A | B_{i+1}) + \\ 
& \frac{1}{2} \cdot \sum_{i = 0, i~odd}^{r-1} I(C_{i+1} ; R, B_{\tau}, O_B, S_B, A_{\tau}, O_A, S_A | A_{i+1}) \\
&= \frac{1}{2} \cdot \sum_{i = 0, i~even}^{r-1} I(C_{i+1} ; O_A | B_{i+1}) + \frac{1}{2} \cdot \sum_{i = 0, i~even}^{r-1} I(C_{i+1} ; R, B_{\tau}, S_B, A_{\tau}, S_A | B_{i+1}, O_A) + \\
& \frac{1}{2} \cdot \sum_{i = 0, i~even}^{r-1} I(C_{i+1} ; O_B | B_{i+1}, R, B_{\tau}, S_B, A_{\tau}, S_A, O_A) + \\
& \frac{1}{2} \cdot \sum_{i = 0, i~odd}^{r-1} I(C_{i+1} ; O_A | A_{i+1}) + \sum_{i = 0, i~odd}^{r-1} I(C_{i+1} ; R, B_{\tau}, S_B, A_{\tau}, S_A | A_{i+1}, O_A) + \stepcounter{equation}\tag{\theequation}\label{myeq1} \\
& \frac{1}{2} \cdot \sum_{i = 0, i~odd}^{r-1} I(C_{i+1} ; O_B | A_{i+1}, R, B_{\tau}, S_B, A_{\tau}, S_A, O_A) \\
&\le \frac{1}{2} \cdot \sum_{i = 0, i~even}^{r-1} I(C_{i+1} ; R, B_{\tau}, S_B, A_{\tau}, S_A | B_{i+1}, O_A) + \\
&\frac{1}{2} \cdot \sum_{i = 0, i~odd}^{r-1} I(C_{i+1} ; R, B_{\tau}, S_B, A_{\tau}, S_A | A_{i+1}, O_A) + O(r) \\
&= \frac{1}{2} \cdot \sum_{i = 0, i~even}^{r-1} \Pr[O_A = 1] \cdot I(C_{i+1} ; R, B_{\tau}, S_B, A_{\tau}, S_A | B_{i+1}, O_A = 1) + \\
&\frac{1}{2} \cdot \sum_{i = 0, i~odd}^{r-1} \Pr[O_A = 1] \cdot I(C_{i+1} ; R, B_{\tau}, S_B, A_{\tau}, S_A | A_{i+1}, O_A = 1) + O(r)
\end{align*}
The first two inequalities are by properties of mutual information. The first equality is just chain rule. Third inequality follows from the fact that $O_A,O_B$ are one dimensional systems. The last equality is true because $O_B$ is just a copy of $O_A$, so tracing out $O_B$, $O_A$ becomes a classical system and also conditioned on $O_A = 0$, the mutual information expressions are $0$ since in that case the $C_{i+1}$ registers are independent of everything else. Now lets analyze the term: 

\begin{align*}
\frac{1}{2} \cdot \sum_{i = 0, i~even}^{r-1} I(C_{i+1} ; R, B_{\tau}, S_B, A_{\tau}, S_A | B_{i+1}, O_A = 1) + \frac{1}{2} \cdot \sum_{i = 0, i~odd}^{r-1} I(C_{i+1} ; R, B_{\tau}, S_B, A_{\tau}, S_A | A_{i+1}, O_A = 1)
\end{align*}
We claim that this is equal to $QIC(\pi^n, \nu')$, where $\nu'$ is the distribution $\nu | O_A = 1$. This follows from the following observations: 
\begin{itemize}
\item Since $O_B$ is just a copy of $O_A$, for all $i$, the state of systems $A_{i+1}, B_{i+1}, C_{i+1}, R, B_{\tau}, S_B, A_{\tau}, S_A$ conditioned on $O_A = 1$ (the post-measurement state if $O_A$ is measured and the result is $1$) is pure. 
\item For all $i$, the marginal state of systems $A_{i+1}, B_{i+1}, C_{i+1}$ conditioned on $O_A = 1$ is the same as it would have been if $\pi^n$ was run starting from the distribution $\nu'$. This is because $\pi^n$ never touches the registers $B_{\tau}, S_B, A_{\tau}, S_A$. 
\item If $\ket{\phi}^{R', A,B,C}$ and $\ket{\psi}^{R,A,B,C}$ are two pure states such that $\text{Tr}_{R'} \ket{\phi}^{R', A,B,C} = \text{Tr}_{R} \ket{\psi}^{R, A,B,C}$. Then $I(C;R'|B)_{\phi} = I(C;R|B)_{\psi}$. 
\end{itemize}
\begin{remark}
The reader might have noticed that the trick of merging stuff with the purification register and then applying the last observation is used at a lot of places in this paper. This seems to be a very useful trick and seems to replace the classical Proposition 2.9 from \cite{B12}. 
\end{remark}

Putting it all together, we have the following upper bound on information cost of step 3: 

\begin{align}
&\Pr[O_A = 1] \cdot QIC(\pi^n, \nu') + O(r) \nonumber \\
&\le \Pr[O_A = 1] \cdot \left( \sum_{i=1}^n QIC(\pi, \nu'_i) \right) + O(r) \nonumber \\
&\le \Pr[O_A = 1] \cdot n \cdot QIC\left(\pi,  \sum_{i=1}^n \nu_i'/n\right) + O(r) \nonumber \\
&\le \Pr[O_A = 1] \cdot n \cdot (I + \delta) + O(\Pr[O_A = 1] \cdot n \cdot r H(w)) + O(r) \label{eqnmain}
\end{align}
Here $\nu'_i$ is the marginal distribution on the $i^{\text{th}}$ coordinate and $w = \sum_{i=1}^n \nu_i'(1,1)/n$. First inequality is by lemma \ref{lem:subadd}. Second inequality is just concavity of information cost, lemma \ref{concavity}. The last inequality follows from corollary \ref{onezero}. Now we can assume that $\Pr[O_A = 1] \ge 1/n$, otherwise (\ref{eqnmain}) is trivially bounded by $O(r)$. Now let us bound $w$. Suppose $(X,Y)$ are random variables s.t. $(X,Y) \sim \nu$. Also let $N(x,y)$ be the number of intersections in $x$ and $y$ i.e. number of $i$ such that $x_i = y_i = 1$. Then 

\begin{align*}
\Pr[N(X,Y) = d | O_A = 1] &= \frac{\Pr[N(X,Y) = d] \cdot \Pr[O_A = 1 | N(X,Y) = d]}{\Pr[O_A = 1]} \\ 
&\le \Pr[N(X,Y) = d] \cdot \Pr[O_A = 1 | N(X,Y) = d] \cdot n \\
&\le  \Pr[N(X,Y) = d] \cdot \left( \left(1-\frac{d}{n}\right)^{n/\log^3(n)}  +  \frac{\log^{30}(n)}{n^{10}}\right) \cdot n \\
&\le e^{-d/\log^3(n)} \cdot n + \frac{\log^{30}(n)}{n^{9}}
\end{align*}

The second inequality follows because if there are $d$ intersections, then getting no intersection in $n/\log^3(n)$ uniformly random coordinates is at most the first term. The second term is due to the error of the amplified protocol for disjointness. So for $d \ge 9 \ln(2) \log^4(n)$, $\Pr[N(X,Y) = d | O_A = 1] \le 1/n^8$. Thus 

\begin{align*}
w =  \sum_{i=1}^n \nu_i'(1,1)/n = \mathbb{E}_{(X,Y) \sim \nu'} N(X,Y)/n \le O(\log^4(n)/n)
\end{align*}
Thus we can bound (\ref{eqnmain}) as follows:

\begin{align*}
 &\Pr[O_A = 1] \cdot n \cdot (I + \delta) + O(\Pr[O_A = 1] \cdot n \cdot r H(w)) + O(r) \\
 &\le n \cdot (I + \delta) + O(n \cdot rH(w)) + O(r) \\
 &\le n \cdot (I+\delta) + O(r \log^5(n))
\end{align*}

Since $\delta$ was arbitrary small, this completes the proof.

\end{proof}

%%%%Everything moved to properties of QIC section
\ignore{

\begin{lemma}
\label{lem:qicvsqccpf}

For any function $f$, error parameter $\epsilon \in [0, 1]$, and number of round $r \in \mathbb{N}$, the following holds:
\begin{align*}
	QIC^r (f, \epsilon) & \leq QCC^r (f, \epsilon), \\
	QIC (f, \epsilon) & \leq QCC (f, \epsilon).
\end{align*}

\end{lemma}

\begin{proof}
%	Given $\delta > 0$, l
Let $\Pi$ be a protocol computing $f$ correctly except with probability $\epsilon$ on 
all input  and satisfying $QCC (\Pi) = QCC (f, \epsilon)$. We get the result by 
noting that $QIC (f, \epsilon) \leq \max_\mu QIC (\Pi, \mu) \leq QCC (\Pi)$.
\end{proof}

\begin{definition}
For all $r \in \mathbb{N}, \epsilon \in [0, 1]$,
\begin{align*}
	QIC_0^r (AND, \epsilon) = \inf_{\Pi \in \T^r (AND, \epsilon)} \max_{\mu_0} QIC (\Pi, \mu_0),
\end{align*}
in which the maximum ranges over all $\mu_0$ satisfying $\mu_0 (1, 1 ) = 0$.
\end{definition}

The following lemma is very similar to Theorem \ref{thm:minimax}. The only difference is that the distributions we consider are restricted and on the right hand side the error of the protocol is measured in the worst case. This lemma will be used in the proof of the next lemma. By applying the same technique as Theorem \ref{thm:minimax}, we get the following lemma.
\begin{lemma}
\label{andminimax}
\[
QIC_0^r (AND, \epsilon) \leq  \max_{\mu_0,\mu_0 (1, 1 ) = 0} \inf_{\Pi \in \T^r (AND, \epsilon)} QIC (\Pi, \mu_0)
\]
\end{lemma}

\begin{lemma}
\label{lem:disjvsand1}
For all $r, n \in \mathbb{N}$,
\begin{align*}
	QCC^r (DISJ_n, 1/3) \geq n \cdot QIC_0^r (AND, 1/3)
\end{align*}
\end{lemma}

\begin{proof}
	We first get the inequality $QCC^r (DISJ_n, 1/3) \geq QIC^r (DISJ_n, 1/3)$ by Lemma~\ref{lem:qicvsqccpf}.
Then, given $\delta > 0$, let $\Pi_D$ be a $r$-round protocol computing $DISJ_n$ 
correctly except with probability $1/3$ on all input and 
satisfying $QIC (\Pi_D, \mu_n) \leq QIC (DISJ_n, 1/3) + \delta$ for 
all $\mu_n$. Let $\Pi_{A,\mu_0}$ be the $r$-round protocol corresponding 
to $\Pi_D$ given by Theorem $4$ in \cite{Dave14}, computing AND 
correctly except with probability 1/3 on all input and 
satisfying 
\[
QIC (\Pi_{A,\mu_0}, \mu_0) = \frac{1}{n} QIC (\Pi_D, \mu_0^{\otimes n}) \leq \frac{1}{n} (QIC (DISJ_n, 1/3) + \delta)
\]
 for each $\mu_0$ with $\mu_0 (1, 1) = 0$. By lemma \ref{andminimax}, there exists a protocol $\Pi_A$ which computes AND 
correctly except with probability 1/3 on all input and 
satisfying 
\[
QIC (\Pi_A, \mu_0) \leq \frac{1}{n} (QIC (DISJ_n, 1/3) + \delta)
\] 
for all $\mu_0$ with $\mu_0 (1, 1) = 0$. Since
\begin{align*}
	QIC_0^r (AND, 1/3)  \leq \max_{\mu_0} QIC (\Pi_A, \mu_0) \leq \frac{1}{n} (QIC (DISJ_n, 1/3) + \delta),\\
\end{align*}
the result follows by taking $\delta$ to $0$.
\end{proof}

\begin{lemma}
	For any function $f$ and error parameter $\epsilon>0$, the following holds:
\begin{align*}
	QIC (f, \epsilon) \in O \big(\log 1/\epsilon \cdot  QIC(f, 1/3)\big).
\end{align*}
\end{lemma}

\begin{proof}
Given $\delta > 0$, let $\Pi$ be a protocol computing $f$ correctly except with 
probability $1/3$ on every input and satisfying $QIC (\Pi, \mu ) \leq QIC(f, 1/3) + \delta$ for 
all $\mu$. Let $n \in O (\log 1/\epsilon)$ be given by the Chernoff bound such that 
protocol $\Pi_n$ running $\Pi$ $n$ times in parallel and taking a majority 
vote (with arbitrary tie-breaking) computes $f$ correctly except with 
probability $\epsilon$ on every input. 
This can be chosen independently of $\delta$.
We now argue on the quantum 
information cost of $\Pi_n$.
Consider an arbitrary distribution $\mu$ for $f$, and let $\mu_n$ be 
the distribution once the $n$ copies have been made.
If we denote the marginal for the $i$-th copy by $\mu^i$, 
then $\mu^i = \mu$. By the proof of Lemma \ref{lem:subadd} and an easy induction, we then get that 
\begin{align*}
	QIC (f, \epsilon) & \leq QIC (\Pi_n, \mu_n) \\
		& \leq n QIC (\Pi, \mu) \\
		&\leq n (QIC (f, 1/3) + \delta).
\end{align*}
 The result follows by taking $\delta$ to $0$.
\end{proof}

The above proof also establishes the following.

\begin{corollary}
\label{cor:reduceAND}
	For any $\epsilon>0$ and $r \in \mathbb{N}$,
\begin{align*}
	QIC_0^r (AND, \epsilon) \in O \big(\log 1/\epsilon \cdot  QIC_0^r (AND, 1/3)\big).
\end{align*}
\end{corollary}

}

\section{Proof of the main result}
%\section{Lower bound on bounded-round AND}
%\jiemingnote{need to change after lemma5.5 is proved, this part should be moved to the end of the paper}
We now put everything together to get a lower bound on $QIC_0^r (AND, 1/3)$.

\begin{lemma}
\label{lem:icand}
	For all $r$, it holds that
\begin{align*}
	QIC_0^r (AND, 1/3) \ge  \Omega\left(\frac{1}{r\cdot \log^8 r} \right).
\end{align*}
\end{lemma}

\begin{proof}
We know by theorem \ref{thm:disj_sherstov}
that $GDM_{1/5} (DISJ_n) \ge \Omega \left(\sqrt{n}\right)$. Hence, by Theorem \ref{lem:QICvsGDM}, we must have 
%for large enough $n$
 that $\max_\mu QIC (DISJ_n, \mu, 2/n) \ge \Omega (\sqrt{n})$. 
Putting this together with Lemma \ref{lem:DISJvsAND2} and 
Corollary~\ref{cor:reduceAND}, and let $r = \Theta \left(\frac{\sqrt{n}}{\log^6 n} \right)$, we have,
\[
QIC_0^r(AND,1/3) = \Omega \left(\frac{1}{\sqrt{n}\cdot \log^2 n} \right) = \Omega \left(\frac{1}{r\cdot \log^8 r} \right).
\]
\end{proof}

\begin{corollary}
\label{cor:icand}
Let $\mu^*$ be the distribution such that $\mu^*(0,0) = 1/3, \mu^*(0,1) = 1/3, \mu^*(1,0) = 1/3$. Then
\[
\inf_{\Pi \in \T^r (AND, 1/3)} QIC (\Pi, \mu^*) =  \Omega \left(\frac{1}{r\cdot \log^8 r} \right).
\]
\end{corollary}

\begin{proof}
For any distribution $\mu_0$ such that $\mu_0(1,1) =0$, it is easy to see that $\mu^*$ can be written as $\mu^* = \frac{1}{3} \mu_0 + \frac{2}{3} \mu'$ where $\mu'$ is some other valid distribution. By Lemma \ref{concavity}, we have 
\[
QIC (\Pi, \mu^*) \geq \frac{1}{3} QIC (\Pi, \mu_0) + \frac{2}{3} QIC (\Pi, \mu') \geq \frac{1}{3} QIC (\Pi, \mu_0).
\]
Then we have 
\[
QIC (\Pi, \mu^*) \geq \frac{1}{3}  \max_{\mu_0,\mu_0(1,1) = 0}QIC (\Pi, \mu_0).
\]
Therefore by Lemma \ref{lem:icand}, we have
\[
\inf_{\Pi \in \T^r (AND, 1/3)} QIC (\Pi, \mu^*) \geq \frac{1}{3} QIC_0^r (AND, 1/3) = \Omega \left(\frac{1}{r\cdot \log^8 r} \right).
\]
\end{proof}

\begin{theorem}
\label{thm:disjround}
For all $r,n \in \mathbb{N}$, $QCC^r (DISJ_n, 1/3) = \Omega \big( \frac{n}{r \cdot \log^8 r}\big)$.
\end{theorem}
\begin{proof}
Combining Lemma \ref{lem:disjvsand1} and Lemma \ref{lem:icand}, we get this theorem.
\end{proof}

\section{Low information protocol for AND}

In this section, we exhibit a $\tilde{O}(1/r)$ information $4r$-round protocol for AND (w.r.t. the prior $1/3,1/3,1/3,0$) which computes correctly on all inputs with probability $1$. The protocol is due to Jain, Radhakrishnan and Sen. Consider the protocol described in Protocol \ref{Protocol:AND}. 

\begin{Protocol}
Inputs: $(x,y) \in \{0,1\} \times \{0,1\}$ \\
Goal: compute $AND(x,y)$
\begin{enumerate}
\item Set $\theta = \frac{\pi}{8r}$. Let $\ket{v}$ be the vector $\cos(\theta) \ket{0} + \sin(\theta) \ket{1}$. Let $U_v$ be the unitary operation of reflecting about the vector $\ket{v}$ i.e. $U_v \ket{0} = \cos(2 \theta) \ket{0} + \sin(2 \theta) \ket{1}$ and $U_v \ket{1} = \sin(2 \theta) \ket{0} - \cos(2 \theta) \ket{1}$. Also let $Z$ be the unitary operation of reflecting about $\ket{0}$ i.e. $Z \ket{0} = \ket{0}$ and $Z \ket{1} = - \ket{1}$. 
\item Alice starts by preparing a qubit $C$ in state $\ket{0}$.
\item If $x = 0$, Alice applies the identity operation on $C$ and sends it to Bob. If $x = 1$, Alice applies the $U_v$ operation on $C$ and sends it to Bob.
\item If $y = 0$, Bob applies the identity operation on $C$ and sends it to Alice. If $y = 1$, Bob applies the $Z$ operation on $C$ and sends it to Alice.
\item After $4r-1$ rounds, Bob measures the register $C$. If the result is $1$, then he answers $1$, otherwise $0$. He also sends this to Alice. 
\end{enumerate}
\caption{Protocol for AND}
\label{Protocol:AND}
\end{Protocol}

First let us see why it computes AND. Let $\ket{\psi^{x,y}_i} = \cos(\phi^{x,y}_i) \ket{0} + \sin(\phi^{x,y}_i) \ket{1}$ be the state of qubit $C$ after $i$ rounds when the input is $(x,y)$. If the input is $0,0$,  $\phi^{0,0}_i$ is always $0$. Also when the input is $0,1$, $\phi^{0,1}_i$ is always $0$. So $\ket{\psi^{0,0}_i} = \ket{\psi^{0,1}_i} = \ket{0}$ always. When the input is $1,0$, $\phi^{1,0}_i$ follows the trajectory $2 \theta \rightarrow 2 \theta \rightarrow 0 \rightarrow 0 \rightarrow 2 \theta \rightarrow \cdots$. So $\ket{\psi^{1,0}_{4r-1}}  = \ket{0}$ as well. When the input is $1,1$, $\phi^{1,1}_i$ follows the trajectory $2 \theta \rightarrow - 2 \theta \rightarrow 4 \theta \rightarrow - 4 \theta \rightarrow \cdots \rightarrow - \pi/2$. So $\ket{\psi^{1,1}_{4r-1}}  = - \ket{1}$. Thus the players compute AND correctly. 

Now let us analyze the information cost of this protocol. Note that after $i$ rounds the full state can be written as follows: 

$$
\ket{\psi_i}^{XYCR} = \sum_{\text{$x,y$ s.t. $x \wedge y = 0$}} \frac{1}{\sqrt{3}} \ket{x}^X \ket{y}^Y \ket{\psi_i^{x,y}}^C \ket{x,y}^R
$$

Then information cost is given by:

\begin{align*}
\frac{1}{2} \cdot \sum_{i = 1, odd}^{4r-1} I(C; R|Y)_{\psi_i} + \frac{1}{2} \cdot \sum_{i = 1, even}^{4r-1} I(C; R|X)_{\psi_i}
\end{align*}

Let us look at a particular term: 

\begin{align*}
I(C; R|Y)_{\psi_i} &= H(C,Y)_{\psi_i} + H(R,Y)_{\psi_i} - H(C,R,Y)_{\psi_i} - H(Y)_{\psi_i} \\
&= H(C,Y)_{\psi_i} + H(C,X)_{\psi_i} - H(X)_{\psi_i} - H(Y)_{\psi_i} \\
&= H(C|Y)_{\psi_i} + H(C|X)_{\psi_i} \\
&= \frac{2}{3} H(C|Y=0)_{\psi_i} + \frac{1}{3} H(C|Y=1)_{\psi_i} + \frac{2}{3} H(C|X=0)_{\psi_i} + \frac{1}{3} H(C|X=1)_{\psi_i} \\
&= \frac{2}{3} H(C|Y=0)_{\psi_i}
\end{align*}
First equality is by definition. For second equality, we are using the fact that for a pure state on some systems $A,B$, $H(A) = H(B)$. Third equality is again by definition. For fourth equality, we use the fact that if we trace out $R$, $X,Y$ become classical. For the fifth equality, we use the fact that conditioned on $Y=1$, system $C$ is in a pure state, namely $\ket{\psi_i^{0,1}}$. Similarly conditioned on $X=1$, it is in state $\ket{\psi^{1,0}_i}$. Conditioned on $X=0$, $C$ is in the state $\ket{0}$. Now conditioned on $Y=0$, $C$ is in the state:

\begin{align*}
\frac{1}{2} \ket{\psi^{0,0}_i} \bra{\psi^{0,0}_i} + \frac{1}{2} \ket{\psi^{1,0}_i} \bra{\psi^{1,0}_i}
\end{align*}
This is $\ket{0}$ if $i \equiv 3 (\text{mod $4$})$ and if $i \equiv 1 (\text{mod $4$})$, the density matrix is given by:

\[
   \rho =
  \left[ {\begin{array}{cc}
   \frac{1}{2} + \frac{1}{2} \cos^2(2 \theta) & \frac{1}{2} \cos(2 \theta) \sin(2 \theta) \\
   \frac{1}{2} \cos(2 \theta) \sin(2 \theta) & \frac{1}{2} \sin^2(2 \theta) \\
  \end{array} } \right]
\]
Eigenvalue computation shows that $H(\rho) = H(\sin^2(\theta)) = O(\theta^2 \log(1/\theta)) = O(\log(r)/r^2)$. So some of Alice's terms are $0$ and some are $O(\log(r)/r^2)$. Similarly some of Bob's terms are $0$ and some are $O(\log(r)/r^2)$. So in total we get that the information cost is $O(\log(r)/r)$. Note that from the protocol it might seem that since the roles of Alice and Bob are asymmetric, only Alice is sending information and Bob is not. However this definition of quantum information cost also accounts for sending back information in some sense. For example, in some of the rounds, Alice is sending Bob some information but Bob is sending it back, so that is accounted for. This results in Bob's part of the cost to be non-zero and in fact equal to that of Alice.

Now let us see what happens if we place a small mass $w$ on $(1,1)$ entry. Then the full state can be described as follows: 

$$
\ket{\psi_i}^{XYCR} = \sum_{\text{$x,y$ s.t. $x \wedge y = 0$}} \sqrt{\frac{1-w}{3}} \ket{x}^X \ket{y}^Y \ket{\psi_i^{x,y}}^C \ket{x,y}^R + \sqrt{w} \ket{1}^X \ket{1}^Y \ket{\psi_i^{1,1}}^C \ket{1,1}^R
$$

\noindent The $i^{\text{th}}$ term of the information cost as before is given by:

\begin{align*}
&\frac{2(1-w)}{3} H(C|Y=0)_{\psi_i} + \frac{1+2w}{3} H(C|Y=1)_{\psi_i} + \frac{2(1-w)}{3} H(C|X=0)_{\psi_i} + \frac{1+2w}{3} H(C|X=1)_{\psi_i} \\
&= \frac{2(1-w)}{3} H(C|Y=0)_{\psi_i} +\frac{1+2w}{3} H(C|Y=1)_{\psi_i} + \frac{1+2w}{3} H(C|X=1)_{\psi_i}
\end{align*}
As before $H(C|X=0)_{\psi_i} = 0$. But the other three terms are non-zero. $H(C|Y=0)_{\psi_i}$ is the same as before. Let us focus on $H(C|Y=1)_{\psi_i}$. State of $C$ conditioned on $Y=1$ is given by:

\begin{align*}
\frac{1-w}{1+2w} \ket{\psi^{0,1}_i} \bra{\psi^{0,1}_i} + \frac{3w}{1+2w} \ket{\psi^{1,1}_i} \bra{\psi^{1,1}_i}
\end{align*}

For $i$ odd, the density matrix is given by:

\[
   \rho =
  \left[ {\begin{array}{cc}
   \frac{1-w}{1+2w} + \frac{3w}{1+2w} \cos^2((i+1) \theta) & \frac{3w}{1+2w} \cos((i+1) \theta) \sin((i+1) \theta) \\
   \frac{3w}{1+2w} \cos((i+1) \theta) \sin((i+1) \theta) & \frac{3w}{1+2w} \sin^2((i+1) \theta) \\
  \end{array} } \right]
\]
Eigenvalue computation shows that 
$$
H(\rho) = H\left(\frac{1 - \sqrt{1 - \frac{12 w (1-w) \sin^2((i+1)\theta)}{(1+2w)^2}}}{2}\right)
$$
Now assuming $w \le 1/6$ and considering $i$ such that $\sin^2((i+1)\theta) \ge 4/5$, we get that 
\begin{align*}
\frac{1 - \sqrt{1 - \frac{12 w (1-w) \sin^2((i+1)\theta)}{(1+2w)^2}}}{2} &\ge \frac{1 - \sqrt{1 - \frac{8w}{(1+2w)^2}}}{2} \\
&= \frac{1 - \frac{1-2w}{1+2w}}{2} \\
&= \frac{2w}{1+2w}
\end{align*}
Since other terms involving $w$ either have positive contribution or are of lower order, we get that for a constant fraction of the rounds, the information cost term increases by an additive $\Omega(H(w))$. And hence overall the increase in information cost is at least $\Omega(r H(w))$. 

\bibliographystyle{alpha}
\bibliography{refs}

\end{document}